\documentclass[letterpaper, 10 pt, conference]{ieeeconf}  

\IEEEoverridecommandlockouts                              
\usepackage{textcomp}

\usepackage{cite}
\usepackage{amsmath,amsfonts,amssymb,graphicx,color,psfrag}
\usepackage{mathrsfs}
\usepackage{graphics}
\usepackage{multirow}
\usepackage{mathtools}
\usepackage{commath}
\usepackage{comment}
\usepackage{caption}
\usepackage{subcaption}
\usepackage{palatino}
\usepackage{bibentry}

\usepackage[pdfencoding=auto, psdextra]{hyperref}
\usepackage{hyperref}
 \hypersetup{
     colorlinks=true,
     linkcolor=blue,
     filecolor=blue,
     citecolor = blue,      
     urlcolor=cyan,
     }

\newtheorem{theorem}{Theorem}[section]
\newtheorem{proposition}[theorem]{Proposition}
\newtheorem{lemma}[theorem]{Lemma}
\newtheorem{remark}[theorem]{Remark}

\newtheorem{definition}{Definition}

\newtheorem{assumption}[theorem]{Assumption}

\newcommand{\subscr}[2]{{#1}_{\textup{#2}}}
\newcommand{\col}[1]{\text{col}\left({#1}\right)}
\DeclareMathOperator{\sign}{sign}

\def\BibTeX{{\rm B\kern-.05em{\sc i\kern-.025em b}\kern-.08em
    T\kern-.1667em\lower.7ex\hbox{E}\kern-.125emX}}
\begin{document}
\title{Comparison of Non-Deterministic Nonlinear Systems}
\author{Shivam Bajaj, Varun S. Madabushi, Maegan Tucker, and Vijay Gupta
}

\maketitle

\begin{abstract}
We characterize a notion of system comparison, termed as $(T_e,\gamma,\delta)$-similarity, for non-deterministic nonlinear systems. Building on a similar notion recently proposed for stable linear systems, the proposed notion characterizes the dissimilarity between the outputs, measured using the $\mathcal{L}_2$ norm, of two nonlinear dynamical systems in terms of their inputs and disturbances. By establishing a relationship between $(T_e,\gamma,\delta)$-similarity and differential dissipativity, we establish equivalence between $(T_e,\gamma,\delta)$-similarity of nonlinear systems and the $(T_e,\gamma,\delta)$-similarity of their differential dynamics. We characterize the $(T_e,\gamma,\delta)$-similarity for nonlinear systems as a Linear Matrix Inequality feasibility problem and also provide necessary and sufficient conditions for solving this feasibility problem. We demonstrate the utility of the proposed notion through its use in two applications: (i) robust hierarchical control applied to a planar aircraft and (ii) the improvement (or design) of abstract models applied to the Moore-Greitzer model and an electronic circuit.
\end{abstract}


\section{Introduction}\label{sec:introduction}
Bisimulation and simulation are central notions of equivalence that have been used extensively for reducing the complexity of dynamical systems. Originating in finite state systems \cite{milner1989communication,park2005concurrency,hennessy1985algebraic}, these notions have also been defined and used for continuous-time dynamical systems~\cite{pappas2002hierarchically,pappas2002consistent,pappas2003bisimilar,lafferriere1997hybrid,van2004equivalence,tabuada2009verification}. Intuitively, simulation characterizes equivalence of one system to another whereas bisimulation characterizes mutual equivalence, i.e., ensuring that two systems are equivalent to each other. Relaxing the equivalence requirement to an approximation requirement, as measured by (a bound over) the maximum distance between the outputs of the two systems, leads to the notion of approximate (bi)simulation~\cite{haghverdi2003bisimulation,girard2011approximate,haghverdi2005bisimulation,prabhakar2018simulations,prabhakar2016bisimulations}.  These notions have become particularly important with the recent interest in applications such as layered or hierarchical control \cite{karimoddini2015hierarchical,kurtz2019formal}, supervisory control \cite{farhat2020control,zhou2006control}, and formal verification \cite{shali2025design,fisler2002bisimulation}.

Early works on approximate (bi)simulation  characterized the bound on the external behaviors using an $\mathcal{L}_{\infty}$ signal norm.  Since one typically enforces equivalence by designing a controller \cite{farhat2020control,zhou2006control}, an equivalence notion that characterizes a bound on the external behaviors using an $\mathcal{L}_{2}$ signal norm is useful to permit the use of control theoretic tools such as dissipativity or robust control \cite{van2000l2,zhou1998essentials}. Recently,~\cite{pirastehzad2024comparison} proposed the notion of $(\gamma,\delta)$-similarity for stable continuous-time non-deterministic linear time-invariant (LTI) dynamical systems using the $\mathcal{L}_{2}$ signal norm. While a precise definition is given in Section \ref{sec:prelim}, informally, two asymptotically stable LTI systems with zero initial conditions are $(\gamma,\delta)$-similar if, for given positive constants $\gamma,\delta$ there exists positive constants $\eta,\mu$, such that for every input $u_1,u_2\in\mathcal{L}_2$ to system $1$ and system $2$, respectively, and disturbance $d_1\in\mathcal{L}_2$ to system $1$, a disturbance $d_2\in\mathcal{L}_2$ to system $2$ such that the outputs $y_1,y_2$ of system $1$ and system $2$ satisfy $\norm{y_1-y_2}^2\leq h_1(\gamma,\delta,\norm{u}_1,\norm{u_2})+h_2(\mu,\eta,\norm{d_1},\norm{d_2}),$ for some non-negative functions $h_i(\cdot)$ which are specified in Section \ref{sec:prelim}. We note here that the notion does not require the inputs $u_1,u_2$ of the two systems to be identical, which allows for compositional reasoning for large-scale interconnected systems. This notion has been used for modular verification \cite{pirastehzad2026modular}, controller synthesis \cite{pirastehzad2026hierarchical}, and abstraction design \cite{pirastehzad2025hierarchical,bajaj2025online} for linear systems.

Given these applications and the desirable properties of $(\gamma,\delta)$-similarity for linear systems combined with the fact that many complex systems are nonlinear, obtaining a similar notion for nonlinear and non-deterministic systems is important. The notion of $(\gamma,\delta)$-similarity is not sufficient for nonlinear systems as it does not account for the presence of finite escape times of nonlinear systems. If system $1$ has a finite escape time, then there may not exist any $d_2\in\mathcal{L}_2$ for which system $2$ is $(\gamma,\delta)$-similar to system $1$ for all $T>0$. Additionally, most frameworks for $H^\infty$ control for nonlinear systems tend to exhibit computational challenges \cite{manchester2018robust}, which is an issue since this machinery ($H^\infty$ control theory) is used to characterize $(\gamma,\delta)$-similarity as an LMI condition. 

%



In this work, we consider two non-deterministic nonlinear systems in continuous time with possibly non-zero initial conditions, each with a finite-escape time greater than a constant $T_e>0$ and introduce the notion of $(T_e,\gamma,\delta)$-similarity that characterizes an upper bound between the external behaviors of two nonlinear dynamical systems for all finite time $T$ satisfying $0<T<T_e$. We note that we do not impose any stability requirement on the systems and are only concerned with the existence of $\gamma$ and $\delta$.  We further use recent developments in contraction theory and feedback design for nonlinear systems \cite{manchester2017control,manchester2018robust,tsukamoto2021contraction} to establish that $(T_e,\gamma,\delta)$-similarity can be characterized via a feasibility problem which is convex in the decision variables and does not require computation of system trajectories. We also characterize a necessary and a sufficient condition for the feasibility problem. 
Additionally, by drawing connections between $(T_e,\gamma,\delta)$-similarity and differential dissipativity,  we establish that $(T_e,\gamma,\delta)$-similarity of nonlinear systems is equivalent to $(T_e,\gamma,\delta)$-similarity of the respective differential dynamics. Accounting for non-determinism in both the systems considered paves the way for applications such as robust hierarchical control of complex systems. 
We demonstrate the applicability of the proposed framework through two such applications. First, we consider the use of the proposed notion of similarity for robust hierarchical control applied to a planar vertical and take off landing (PVTOL) aircraft, where we first determine the control input and consequently a reference trajectory using an abstract similar model and then design a reference tracking controller for the underlying PVTOL. While this idea has been explored in the past \cite{li2021model,kim2023layered,karimoddini2015hierarchical,kurtz2019formal}, using  $(T_e,\gamma,\delta)$-similarity offers multiple advantages such as an analytical bound on the error between the outputs, disturbance rejection, and accounting for actuation errors. In the second application, we use the proposed notion for the improvement (or design) of existing abstract models applied to the Moore-Greitzer model and an electronic circuit model.

The rest of this work is organized as follows. We present the notion of $(T_e,\gamma,\delta)$-similarity for nonlinear systems along with an LMI characterization in Section~\ref{sec:definition}. 
We establish
necessary and sufficient conditions that are simpler to check in Section \ref{sec:conditions}. In Section \ref{sec:properties}, we introduce the notion of differential $(T_e,\gamma,\delta)$-similarity and study the relation between the similarity of nonlinear systems and their differential dynamics. In Section \ref{sec:numerics}, we propose two applications and present numerical results for these applications. Finally, we conclude this work and discuss future works in Section \ref{sec:conclusion}.

\paragraph*{Notation} We use $\abs{x}$ to denote the standard Euclidean norm of a real vector $x$. We define the operator $\col{\cdot,\cdot}$ such that $\col{x_1,x_2} = \begin{bmatrix}
    x_1^\top, x_2^\top
\end{bmatrix}^\top$ for any vectors $x_1\in \mathbb{R}^{n_1}$ and $x_2\in \mathbb{R}^{n_2}$.
For a given $T$, we denote $\mathcal{L}_{2,T}$ as the space of square-integrable vector signal on the interval $[0,T]$.
The associated squared-norm of a signal $x\in \mathcal{L}_{2,T}$ is denoted as $\norm{x}^2_{[0,T]}\coloneqq \int_0^T \abs{x(t)}^2 dt < \infty$.
Given a positive definite matrix $M$, we define $\norm{x}_{M,T}^2 = \int_0^T x(t)^\top Mx(t) dt $. We use $\mathbf{I}$ and $\mathbf{0}$ to denote identity matrix and zero matrix of appropriate dimensions, respectively. A Riemannian metric on $\mathbb{R}^n$ is a symmetric positive-definite matrix function, smooth on $x$ and denoted as $M(x)$, which defines a Euclidean structure via an inner product $\langle\Delta_1,\Delta_2\rangle_x = \Delta_1^\top M(x)\Delta_2$, where $\Delta_1$ and $\Delta_2$ are any two tangent vectors. A metric is called uniformly bounded if, for some positive constants $\alpha_2\geq \alpha_1$, $\alpha_1 \mathbf{I}\leq M(x) \leq \alpha_2 \mathbf{I}, \forall x$. 
Let $\Gamma(a,b)$ be the set of smooth paths between two points $a$ and $b$ in $\mathbb{R}^n$, where each $c\in \Gamma(a,b)$ is a piecewise smooth mapping $c:[0,1]\to \mathbb{R}^n$ and $c(0) = a$ and $c(1) = b$. For a given metric $M(x)$, the \emph{energy} of a path is defined as
\begin{align*}
    E(c) \coloneqq \int_0^1 \left(\frac{\partial c(s)}{\partial s}\right)^\top M\left(c(s)\right)\frac{\partial c(s)}{\partial s} ds.
\end{align*}
Further, we use $E(a,b)\coloneqq \inf_{c\in \Gamma(a,b)} E(c)$ for $a,b\in \mathbb{R}^n$ to denote the minimal energy of a path joining $a$ and $b$.

\section{\texorpdfstring{$(T_e,\gamma,\delta)$}{(Te,\gamma,\delta)}-Similarity for Nonlinear Systems}
\label{sec:definition}

In this section, formally define $(T_e,\gamma,\delta)$-similarity for nonlinear systems and present an LMI characterization for the same. We will first present the definition of $(\gamma,\delta)$-similar systems for stable LTI systems which may be useful for readers unfamiliar with \cite{pirastehzad2024comparison}. Then we will present the notion of $(T_e,\gamma,\delta)$-similarity for nonlinear systems followed by its LMI characterization.  

\subsection{A Refresher on \texorpdfstring{$(\gamma,\delta)$}{(\gamma,\delta)}-similarity for Stable Continuous-Time LTI systems}\label{sec:prelim}
In this section, we provide the definition of $(\gamma,\delta)$-similarity for stable continuous-time LTI systems. 
While the definitions provided in this section are not instrumental for this work, we provide them for readers who are unfamiliar with the notion of $(\gamma,\delta)$-similarity. We refer to \cite{pirastehzad2024comparison} for additional details.

For $i\in\{1,2\}$, consider two continuous-time $0$-asymptotically stable linear systems  of the form:

\begin{equation}\label{eq:dynamics_prelim}
    \Sigma_i \colon \begin{cases}
        \dot{x_i} &= A_ix_i + B_iu_i + E_id_i; x_i(0) = 0\\
        y_i &= C_ix_i, 
    \end{cases}
\end{equation}

where $x_i \in \mathbb{R}^{n_i}$, $u_i\in \mathbb{R}^{m}$, $d_i\in \mathbb{R}^{q_i}$, and $y_i\in\mathbb{R}^{p}$ represent the state, input, disturbance, and output, respectively, of system $\Sigma_i$.

\begin{definition}[$(\gamma,\delta)$-similarity]\label{def:similar_systems}
   Given $0$-asymptotically stable systems $\Sigma_i$, $i\in\{1,2\}$, and positive constants $\gamma$ and $\delta$, system $\Sigma_2$ is said to be $(\gamma,\delta)$-similar to $\Sigma_1$, if there exist constants $\epsilon,\eta>0,\mu>0$ such that for every input $u_1,u_2\in \mathcal{L}_2$ and disturbance $d_1\in\mathcal{L}_2$, there exists a disturbance $u_d\in\mathcal{L}_2$ such that
   \begin{equation}\label{eq:similarity}
   \begin{split}
       \norm{y_1-y_2}^2\leq & \gamma\norm{u_1-u_2}^2 + (\delta-\epsilon)\norm{\col{u_1,u_2}}^2\\
       &+ (\mu-\epsilon)\norm{d_1}^2 - \eta\norm{u_d}^2.
   \end{split}
   \end{equation}
\end{definition}

Next, we present the finite-time analogue of Definition \ref{def:similar_systems}.

\begin{definition}[Finite-time $(\gamma,\delta)$-similarity]\label{def:similar_systems_finite}
   Given $0$-asymptotically stable systems $\Sigma_i$, $i\in\{1,2\}$, a terminal time $T>0$, and positive constants $\gamma$ and $\delta$, system $\Sigma_2$ is said to be $(\gamma,\delta)$-similar to $\Sigma_1$ in the interval $[0,T]$, if there exist constants $\epsilon,\eta>0,\mu>0$ such that for every input $u_1,u_2\in \mathcal{L}_{2,T}$ and disturbance $d_1\in\mathcal{L}_{2,T}$, there exists a disturbance $u_d\in\mathcal{L}_{2,T}$ such that
   \begin{equation}\label{eq:similarity_finite}
   \begin{split}
       \norm{y_1-y_2}^2_{[0,T]}\leq & \gamma\norm{u_1-u_2}^2_{[0,T]} + (\delta-\epsilon)\norm{\col{u_1,u_2}}^2_{[0,T]}\\
       &+ (\mu-\epsilon)\norm{d_1}^2_{[0,T]} - \eta\norm{u_d}^2_{[0,T]}.
   \end{split}
   \end{equation}
\end{definition}
The notion of $(\gamma,\delta)$-similarity measures the similarity of the trajectories of $\Sigma_1$ and $\Sigma_2$ in terms of their input-output behavior. In other words, it measures to what extend is the behavior of $\Sigma_1$ is contained in $\Sigma_2$.
While the parameter $\epsilon$ is for technical purposes, the parameter $\gamma$ characterizes the deviation in the outputs with respect to the dissimilarity in the inputs.  The parameter $\delta$ characterizes the deviation in the outputs with respect to the individual inputs. Finally, the constants $\mu$ and $\eta$ characterize the deviation in the output with respect to disturbances.

\subsection{System Model and Definition}
Consider two continuous time nonlinear systems $\Sigma_1$ and $\Sigma_2$ of the following form:

\begin{equation}\label{eq:dynamics_1}
    \Sigma_1: \begin{cases}
        \dot{x}_1 = f_1(x_1,u_1,d_1), \quad x_1(0) = x_{1,0};\\
        y_1 = g_1(x_1),
    \end{cases}
\end{equation}

\begin{equation}\label{eq:dynamics_2}
    \Sigma_2: \begin{cases}
        \dot{x}_2 = f_2(x_2,u_2,d_2,u_d), \quad x_2(0) = x_{2,0};\\
        y_2 = g_2(x_2),
    \end{cases}
\end{equation}
where, for $i\in\{1,2\}$, $x_i \in \mathbb{R}^{n_i}$ denotes the states, $u_i\in \mathbb{R}^m$ denotes the control input, $d_i\in \mathbb{R}^{q_i}$ denotes arbitrary unknown disturbance, and $y_i\in\mathbb{R}^p$ denotes the output. The term $u_d\in\mathbb{R}^{q_d}$ denotes an additional external control (referred to as driving input) on $\Sigma_2$. We assume that $f_{i}$ and $g_i$, $i\in\{1,2\}$, are in $\mathcal{C}^1$ and are time-invariant.
For system $\Sigma_1$ (resp. $\Sigma_2$), we use $T_e^1 >0$ (resp. $T_e^2>0$) to denote the finite escape time.

We compare system $\Sigma_1$ and system $\Sigma_2$ characterized in \eqref{eq:dynamics_1} and \eqref{eq:dynamics_2} according to the following definition.

\begin{definition}[$(T_e,\gamma,\delta)$-similarity]\label{def:similarity}
    Given systems $\Sigma_1$ and $\Sigma_2$ of the form in \eqref{eq:dynamics_1} and \eqref{eq:dynamics_2}, respectively,  system $\Sigma_2$ is said to be $(T_e,\gamma,\delta)$-similar to $\Sigma_1$ if there exist positive constants $\gamma, \delta, \mu$, and $\eta$ such that for every  $u_1, u_2, d_1, d_2 \in \mathcal{L}_{2,T}$ and every $x_{i,0}$, $i\in \{1,2\}$, there exist a driving input $u_d\in \mathcal{L}_{2,T}$ such that, for all finite $T$ satisfying $0<T<T_e$,
    \begin{equation}\label{eq:similarity_def}
    \begin{split}
        \norm{y_1-y_2}^2_{[0,T]} &\leq \gamma\norm{u_1-u_2}^2_{[0,T]} + \delta \norm{\col{u_1,u_2}}^2_{[0,T]} \\
        &+ \mu\norm{\col{d_1,d_2}}^2_{[0,T]} -\eta\norm{u_d}^2_{[0,T]} \\
        &+ b\left(x_{1,0},x_{2,0}\right),
    \end{split}
    \end{equation}
    for some  finite-valued continuous function $b(x_{1,0},x_{2,0})\geq 0$ that satisfies $b(x,x)=0$.
\end{definition}

For well-posedness of \eqref{eq:similarity_def}, we assume that the initial conditions $x_0 \coloneqq \col{x_{1,0},x_{2,0}}$ and the signal $w \coloneqq \col{u_1, u_2, d_1, d_2}$ are such that $b(x_0)=0$ and $w=0$ do not jointly hold.

The notion of $(T_e,\gamma,\delta)$-similarity characterizes the similarity of $\Sigma_1$ and $\Sigma_2$ in terms of their input-output behavior.  The constant $\gamma$ characterizes the effect of dissimilar inputs on the error between the outputs of the two systems and the constant $\delta$ characterizes the effect of each input on the error, which is essential due to the presence of actuation errors. Definition \ref{def:similarity} captures the case when $u_1 = u_2$. Finally, the constant $\mu$ characterizes the effect of disturbance $d_1$ and $d_2$ on the error between the outputs of the two systems. The notion of $(T_e\gamma,\delta)$-similarity characterizes to what extent the input–output behavior of $\Sigma_1$ is contained in $\Sigma_2$.

Apart from the fact that Definition \ref{def:similarity} considers nonlinear systems, Definition \ref{def:similarity} differs from the definition of $(\gamma,\delta)$-similarity in \cite{pirastehzad2024comparison} in multiple ways.
Definition \ref{def:similarity} generalizes Definition \ref{def:similar_systems_finite} by accounting for non-zero initial conditions as well as non-determinism in both system $\Sigma_1$ and system $\Sigma_2$. 
In other words, for a fixed $T$, when system $\Sigma_1$ and $\Sigma_2$ are linear time invariant with $b(x_{1,0},x_{2,0})=0$ and $d_2(t)=0$, Definition \ref{def:similarity} reduces to Definition \ref{def:similar_systems_finite}. 
In Section \ref{sec:numerics}, we show that accounting for non-determinism in both systems allows one to use this framework for robust hierarchical control of complex systems.
Additionally, Definition \ref{def:similar_systems_finite} requires system $\Sigma_1$ and $\Sigma_2$ to be $0$-asymptotically stable and also requires constants $\gamma$ and $\delta$ to be given. In contrast, Definition \ref{def:similarity} does not impose any stability requirement and is only concerned with the existence of $\gamma$ and $\delta$.  
We also note that Definition \ref{def:similarity} considers signal norms in the $\mathcal{L}_{2,T}$ space in contrast to $\mathcal{L}_2$ in Definition \ref{def:similar_systems_finite}. 
This distinction is primarily due to the fact that nonlinear systems have finite escape times and that the solution of the nonlinear system may not exist. 
Under the assumption that system $\Sigma_1$ is asymptotically stable and that the solution for both system $\Sigma_1$ and $\Sigma_2$ exist for all $T>0$, it is possible to extend this work to when $u_1,u_2,d_1,d_2 \in \mathcal{L}_2$. 

Definition \ref{def:similarity} characterizes a notion of similarity between two systems. However, using Definition \ref{def:similarity} directly to check for $(T_e,\gamma,\delta)$-similarity is challenging.
Thus, a verifiable condition for $(T_e,\gamma,\delta)$-similarity that does not require computing the trajectories of the two systems is desirable. We provide such a characterization in the next subsection by considering an alternate formulation.

\subsection{LMI Characterization of \texorpdfstring{$(T_e,\gamma,\delta)$}{(Te,\gamma,\delta)}-Similarity}\label{sec:characterization}
Let

\begin{equation}\label{eq:Q_and_R}
    Q\coloneqq \begin{bmatrix}
    (\gamma+\delta)\mathbf{I} & -\gamma\mathbf{I} & \mathbf{0} & \mathbf{0}\\
    -\gamma\mathbf{I} & (\gamma+\delta)\mathbf{I} & \mathbf{0} & \mathbf{0}\\
    \mathbf{0} & \mathbf{0} & \mu\mathbf{I} & \mathbf{0} \\
    \mathbf{0} & \mathbf{0} & \mathbf{0} & \mu\mathbf{I}
    \end{bmatrix}, \quad
    R \coloneqq \begin{bmatrix}
        \mathbf{I} & \mathbf{0}\\
        \mathbf{0} & \eta\mathbf{I}
    \end{bmatrix}.
\end{equation}

Further, let $x\coloneqq \col{x_1,x_2}$, $w\coloneqq \col{u_1,u_2,d_1,d_2}$, and $z\coloneqq \col{y_1-y_2,d_2}$. Consider the composite system
\begin{equation}\label{eq:dynamics_composite}
    \Sigma: \begin{cases}
        \dot{x} = f(x,w,u_d), \quad x(0) = x_0;\\
        z = g(x,u_d),
    \end{cases}
\end{equation}
where 
\begin{align*}
    f(x,w,u_d) &\coloneqq \begin{bmatrix}
        f_1(x_1,u_1,d_1)\\
        f_2(x_2,u_2,d_2,u_d)
    \end{bmatrix},
    z(x,u_d) \coloneqq
    \begin{bmatrix}
            y_1 - y_2\\
            u_d
    \end{bmatrix},\\
    g(x,u_d) &=\begin{bmatrix}
        g_1(x_1)-g_2(x_2)\\
        u_d
    \end{bmatrix}.
\end{align*}
Then we can express \eqref{eq:similarity_def} equivalently as the following:

\begin{equation}\label{eq:similarity_comp}
    \norm{z}_{R,T}^2 \leq \norm{w}_{Q,T}^2 + b\left(x_{0}\right).
\end{equation}


As we will see shortly, the condition for $(T_e,\gamma,\delta)$-similarity can be characterized in terms of the differential dynamics of $\Sigma_i$, $i\in\{1,2\}$, and the matrices $Q$ and $R$. To this end, the differential dynamics of $\Sigma_1$ of the form in \eqref{eq:dynamics_1} is:

\begin{align}\label{eq:diff_dyn_1}
    \Delta\Sigma_1: \begin{cases}
        \Delta\dot{x}_1 = A_1\Delta x_1 + B_1\Delta u_1 + E_1\Delta d_1\\
        \Delta y_1 = C_1\Delta x_1,
    \end{cases}
\end{align}
where $A_1(x_1,u_1,d_1) = \frac{\partial f_1}{\partial x_1}$, $B_1(x_1,u_1,d_1) = \frac{\partial f_1}{\partial u_1}$, $E_1(x_1,u_1,d_1) = \frac{\partial f_1}{\partial d_1}$, and $C_1(x_1) = \frac{\partial g_1}{\partial x_1}$. 
Similarly, the differential dynamics of $\Sigma_2$ of the form in \eqref{eq:dynamics_2} is:

\begin{align}\label{eq:diff_dyn_2}
    \Delta\Sigma_2: \begin{cases}
        \Delta\dot{x}_2 = A_2\Delta x_2 + B_2\Delta u_2 + E_2\Delta d_2+E_d\Delta u_d\\
        \Delta y_2 = C_2\Delta x_2,
    \end{cases}
\end{align}
where $A_2(x_2,u_2,d_2,u_d) = \frac{\partial f_2}{\partial x_2}$,  $B_2(x_2,u_2,d_2,u_d) = \frac{\partial f_2}{\partial u_2}$, $E_2(x_2,u_2,d_2,u_d) = \frac{\partial f_2}{\partial d_2}$, $E_d(x_2,u_2,d_2,u_d) = \frac{\partial f_2}{\partial u_d}$, and $C_2(x_2) = \frac{\partial g_2}{\partial x_2}$. 

Analogous to that for systems $\Sigma_i$, $i\in \{1,2\}$, the composite system of differential dynamics $\Delta\Sigma_i$ is:

\begin{align}\label{eq:diff_dyn_composite}
    \Delta\Sigma: \begin{cases}
        \Delta\dot{x} = A\Delta x + B\Delta w + E\Delta u_d\\
        \Delta z = C\Delta x + D\Delta u_d,
    \end{cases}
\end{align}
where 
\begin{subequations}
\begin{align}
    \begin{split}
        A &= \begin{bmatrix}
        A_1 & \mathbf{0}\\
        \mathbf{0} & A_2
    \end{bmatrix}=  \frac{\partial f}{\partial x},
    \end{split}\\
    \begin{split}
        B &= \begin{bmatrix}
        B_1 & \mathbf{0} & E_1 & \mathbf{0}\\
        \mathbf{0} & B_2 & \mathbf{0} & E_2
    \end{bmatrix} = \frac{\partial f}{\partial w},
    \end{split}\\
    \begin{split}\label{eq:E}
        E  &= \begin{bmatrix}
        \mathbf{0}\\
        E_d
    \end{bmatrix} = \frac{\partial f}{\partial d_2},
    \end{split}\\
    \begin{split}
        C  &=  \begin{bmatrix}
        C_1 & -C_2\\
        \mathbf{0} & \mathbf{0}
    \end{bmatrix} = \frac{\partial z}{\partial x},
    \end{split}\\
    \begin{split}
        D  = \begin{bmatrix}
        \mathbf{0}\\
        \mathbf{I}
    \end{bmatrix} = \frac{\partial z}{\partial u_d}.
    \end{split}
\end{align}
\end{subequations}

To reduce notational overload, we do not explicitly show the dependence on the $x(t), w(t), u_d(t)$ in the notation of system matrices of $\Delta \Sigma$.
The differential dynamics $\Delta\Sigma_i$ are linear time-varying and are evaluated at particular trajectories of $x(t),w(t),u_d(t)$ of system \eqref{eq:dynamics_composite}. For the remainder of this work, we assume the following.

\begin{assumption}\label{assum:d_2_feedback}
    If $\Sigma_2$ is $(T_e,\gamma,\delta)$-similar to $\Sigma_1$, then the driving input $u_d(t)$ follows a feedback law, i.e.,
    \begin{equation}\label{eq:dist_form}
    u_d(t) = k(x(t)).
\end{equation}
\end{assumption}

Assumption \ref{assum:d_2_feedback} allows us to characterize the notion of $(T_e,\gamma,\delta)$-similarity as a verifiable condition that depends on the system matrices of $\Delta \Sigma$. 
It is possible that, for some $(T_e,\gamma,\delta)$-similar systems, no $u_d(t)$ of the form \eqref{eq:d_2_form} exist. 
Assumption \ref{assum:d_2_feedback} implies that $u_d(t)$ is a function of the state of both system $\Sigma_1$ and system $\Sigma_2$. 
The requirement on the state information is not restrictive as there are applications, such as hierarchical control wherein the abstract model is a \emph{simulated model}, in which the state of the simulated model is always available.
The requirement on the state information is also unavoidable as we do not impose any restriction on the external disturbance $d_1(t)$. 
For the case of LTI systems \cite{pirastehzad2024comparison}, it was formally established that $u_d(t)$ must be of the state-feedback form requiring the state information from both systems.

Selecting $u_d(t)$ of the form \eqref{eq:dist_form} implies that system $\Sigma_2$ is of the form (cf. equation \eqref{eq:dynamics_2})
\begin{align*}
    \dot{x}_2 &= f_2(x_1,x_2,u_2,d_2), \quad x_2(0) = x_{2,0};\\
        y_2 &= g_2(x_2)
\end{align*}
with differential dynamics as
\begin{equation}\label{eq:diff_dyn_cl_2}
    \subscr{\Delta\Sigma}{2,cl}: \begin{cases}
        \Delta\dot{x}_2 = A_2\Delta x_2 + B_2\Delta u_2 + E_2\Delta d_2 + E_dK\Delta x\\
        \Delta y_2 = C_2\Delta x_2,
    \end{cases}
\end{equation}
where $K(x) = \frac{\partial k(x)}{\partial x} = \begin{bmatrix}
    K_1 & K_2
\end{bmatrix}$.
The differential dynamics of the closed-loop composite system can be expressed as
\begin{equation}\label{eq:diff_dyn_cl}
    \subscr{\Delta\Sigma}{cl}: \begin{cases}
        \Delta\dot{x} = \mathcal{A}\Delta x + B\Delta w\\
        \Delta z = \mathcal{C}\Delta x,
    \end{cases}
\end{equation}
where $\mathcal{A} = \left(A+EK\right),$ and $\mathcal{C} = \left(C + DK\right)$.

We are now ready to present the main result of this section.
The central idea is to select a driving input $u_d(t)$ such that, when applied to the composite system $\eqref{eq:dynamics_composite}$, the path of minimum energy joining $x(t)$ to $x^*(t)\coloneqq \col{x_1(t),x_1(t)}$ is minimized. 
We assume the following:
\begin{assumption}\label{assum:zero_measure}
 Let $x^*(t)\coloneqq \col{x_1(t),x_1(t)}$. For the composite system \eqref{eq:dynamics_composite} and its differential dynamics \eqref{eq:diff_dyn_composite}, the set of times $t\in [0,T)$ for which $x(t)$ is in the cut locus of $x^*(t)$ has zero measure.
\end{assumption}


We now present an LMI characterization of $(T_e,\gamma,\delta)$-similarity (cf. Definition \ref{def:similarity}) in the following result.

\begin{theorem}\label{thm:LMI_characterization}
    Suppose Assumptions \ref{assum:d_2_feedback} and \ref{assum:zero_measure} hold. Further suppose that there exists a uniformly-bounded dual metric $W(x)$ and a matrix function $Y(x)$ that satisfy 
    \begin{equation}\label{eq:LMI_W}
        \begin{bmatrix}
            \mathcal{W} & \left(CW+DY\right)^\top & B \\
            CW+DY & R^{-1} & \mathbf{0}\\
            B^\top & \mathbf{0} & Q
        \end{bmatrix} \geq 0,
    \end{equation}
    where $\mathcal{W} = \dot{W} -WA^\top -AW-EY -Y^\top E^\top$. Then, system $\Sigma_2$ is $(T_e,\gamma,\delta)$-similar to system $\Sigma_1$.
\end{theorem}
\begin{proof}
    Consider $u_d(t)$ of the following form:

    \begin{subequations}\label{eq:d_2_form}
        \begin{equation}\label{eq:nu}
            \nu(t) = \arg \min_{c \in \Gamma\left(x(t),x^*(t)\right)} E(c) 
        \end{equation}
        \begin{equation}\label{eq:d_2}
            u_d(t) = \int_0^1 K\left(\nu(t,s)\right)\frac{\partial \nu}{\partial s} ds,
        \end{equation}
    \end{subequations}
    with $K(x) \coloneqq Y(x)W^{-1}(x)$. The proof for the existence of such a $u_d(t)$ is analogous to that of \cite[Theorem 1]{manchester2017control} and has been omitted for brevity.

    At any time $\tau \in [0,\infty)$, consider the following smoothly parameterized paths of states $c(t,s)$, control $d(t,s)$, disturbance $\omega(t,s)$, and output $\zeta(t,s)$ for $s\in [0,1]$:
    \begin{itemize}
        \item states $c(t,s) \coloneqq \nu(t,s)$,
        \item control $d(t,s)$ determined as:
        \begin{align}
            d(t,s) \coloneqq \int_0^s K\left(c\left(t,\tilde{s}\right)\right)\frac{\partial c\left(t,\tilde{s}\right)}{\partial \tilde{s}}d\tilde{s},
        \end{align}
        \item disturbance $\omega(t,s)\coloneqq sw(t)$,
        \item outputs $\zeta(t,s) = g\left(c\left(t,s\right),d\left(t,s\right)\right)$.
    \end{itemize}

    Differentiating each of the above four paths with respect to $s$ at fixed time $t = \tau$
    \begin{subequations}\label{eq:path_diff}
    \begin{align}
        \frac{\partial c\left(t,s\right)}{\partial s} &= \frac{\partial \nu\left(t,s\right)}{\partial s},\\
        \frac{\partial d\left(t,s\right)}{\partial s} &= K\left(\nu\left(t,s\right) \right)\frac{\partial \nu\left(t,s\right)}{\partial s},\\
        \frac{\partial \omega\left(t,s\right)}{\partial s} &= w(t),\\
    \begin{split}
        \frac{\partial \zeta\left(t,s\right)}{\partial s} &= C\left(\nu\left(t,s\right),d\left(t,s\right)\right)\frac{\partial \nu\left(t,s\right)}{\partial s} \\
        &+ D\left(\nu\left(t,s\right),d\left(t,s\right)\right)\frac{\partial d\left(t,s\right)}{\partial s}. 
    \end{split}
    \end{align}
    \end{subequations}

    Let $\phi>0$ be some arbitrary small constant. For time interval $[\tau,\tau+\phi]$ and for each $s\in [0,1]$, fix $d(t,s)$ and $\omega(t,s)$ to their values at $t=\tau$. Then, by interchanging the order of differentiation of $t$ and $s$, it follows that the path derivatives $\frac{\partial c\left(t,s\right)}{\partial s}, \frac{\partial d\left(t,s\right)}{\partial s},\frac{\partial \omega\left(t,s\right)}{\partial s},$ and $\frac{\partial \zeta\left(t,s\right)}{\partial s}$ satisfy \eqref{eq:diff_dyn_cl} with the substitution $\Delta x=\frac{\partial c\left(t,s\right)}{\partial s}, \Delta u_d =\frac{\partial d\left(t,s\right)}{\partial s}, \Delta w = \frac{\partial \omega\left(t,s\right)}{\partial s},$ and $\Delta z=\frac{\partial \zeta\left(t,s\right)}{\partial s}$.

    Let $M(x) = W^{-1}(x)$. Then, by application of Schur's complement, equation \eqref{eq:LMI_W} is equivalent to 
    \begin{equation}\label{eq:LMI_M}
        \begin{bmatrix}
            \dot{M} + M\mathcal{A} + \mathcal{A}^\top M + \mathcal{C}^\top R \mathcal{C} & MB \\
            B^\top M & -Q
        \end{bmatrix} \leq 0.
    \end{equation}
    This implies
    \begin{align*}
        \frac{d}{dt}\left(\Delta x^\top M\Delta x\right) \leq -\Delta z^\top R \Delta z + \Delta w^\top Q \Delta w.
    \end{align*}

    Integrating over $s\in [0,1]$ and
    yields
    \begin{align}\label{eq:interim}
        \int_0^1 \frac{d}{dt}\left(\Delta x^\top M\Delta x\right) ds \leq & -\int_0^1 \Delta z^\top R \Delta z ds \nonumber\\
        &+ \int_0^1\Delta w^\top Q \Delta w ds.
    \end{align}
    Noting from equation \eqref{eq:path_diff} that $\Delta w = \frac{\partial \omega\left(t,s\right)}{\partial s} = w(t)$ and that $z^\top Rz = \left(\int_0^1 \frac{\partial \zeta\left(t,s\right)}{\partial s}ds\right)^\top R\left(\int_0^1 \frac{\partial \zeta\left(t,s\right)}{\partial s}ds\right)\leq \int_0^1  \left(\frac{\partial \zeta\left(t,s\right)}{\partial s} \right)^\top R\left( \frac{\partial \zeta\left(t,s\right)}{\partial s}\right)ds$ (from Cauchy-Schwarz inequality) yields
    \begin{align}\label{eq:interim_1}
        \int_0^1 \frac{d}{dt}\left(\Delta x^\top M\Delta x\right) ds \leq &  -z^\top R z +  w^\top Q w.
    \end{align}

    Given the assumption that the dynamics $f$ and perturbation input $u_d$ are smooth, for all sufficiently small $\epsilon>0$, there exists a path $c(t)$ with 
    \begin{align*}
        E\left(c(\tau+\epsilon)\right) - E\left(c(\tau)\right)\leq \int_{\tau}^{\tau+\epsilon} \left(-z^\top R z +  w^\top Q w\right)dt.
    \end{align*}

    Since $E\left(x(t),x^*(t)\right)$ is the minimum energy of a path, $E\left(x(t),x^*(t)\right) \leq E\left(c(t)\right)$ for $t\in [\tau,\tau+\epsilon)$. Further, $E\left(x(t),x^*(t)\right) = E\left(c(t)\right)$ on $t=\tau$. Since $\tau$ is arbitrary and taking $\epsilon\to 0$,
    it follows that for all $t$
    \begin{align*}
        \frac{d}{dt}E\left(x(t),x^*(t)\right) \leq -z^\top R z +  w^\top Q w.
    \end{align*}
    Integrating over time and using that $E\left(x(t),x^*(t)\right)\geq 0$ establishes the claim with $b(x,x^*) = E\left(x(t),x^*(t)\right)$.

    We now establish that $u_d(t) \in \mathcal{L}_{2,T}$. 
    For any $0<T<T_e$, since $Y(x)$ is smooth and $M(x)=W^{-1}(x)$ is a uniformly bounded dual metric, it follows that there exists a constant $\bar{K}<\infty$ such that $K(x)\leq \bar{K}$. 
    This implies
    \begin{align*}
    \abs{u_d(t)} &\leq \bar{K} \int_0^1 \left(\frac{\partial \nu}{\partial s}\right)^\top  \left(\frac{\partial \nu}{\partial s}\right)ds,\\
        & \leq \frac{\bar{K}}{\alpha_1} \int_0^1 \left(\frac{\partial \nu}{\partial s}\right)^\top M \left(\frac{\partial \nu}{\partial s}\right)ds = \frac{\bar{K}}{\alpha_1} E\left(\nu(t) \right).
    \end{align*}
    where for the second inequality we used the fact that $M(x)$ is uniformly bounded, i.e., $\alpha_1 \mathbf{I}\leq M(x) \leq \alpha_2\mathbf{I}$ for some $\alpha_2\geq \alpha_1> 0$.

    Interchanging the differentiation and integration in equation \eqref{eq:interim_1} and since $\Delta x = \frac{\partial c}{\partial s} = \frac{\partial \nu}{\partial s}$, we obtain

    \begin{align*}
        \frac{d}{dt} E(\nu(t)) \leq & \int_0^1 w^\top Q w ds
        \implies E(\nu(t)) \leq \norm{w}_{Q,T}^2.
    \end{align*}
    Further, since $w\in \mathcal{L}_{2,T}$, there exists a constant $\bar{w}>0$ such that $\norm{w}_{Q,T}^2\leq \bar{w}$. Then, for any $T>0$
    
    \begin{align*}
        \abs{u_d(t)} \leq \frac{\bar{K}}{\alpha_1}\norm{w}_{Q,T}^2\leq \frac{\bar{K}\bar{w}}{\alpha_1}\\
        \implies \norm{u_d(t)} \leq \frac{T\bar{K}\bar{w}}{\alpha_1} < \infty.
    \end{align*}
    This concludes the proof.
\end{proof}

Theorem \ref{thm:LMI_characterization} provides a sufficient condition for $(T_e,\gamma,\delta)$-similarity for nonlinear systems in terms of a feasibility problem of a pointwise LMI in $W$ and $Y$. In other words, \eqref{eq:LMI_W} is convex but infinite dimensional. One may use tools such as Sum-of-Squares programming to solve equation \eqref{eq:LMI_W} efficiently.
Further, Theorem \ref{thm:LMI_characterization} requires Assumption \ref{assum:zero_measure} which can be challenging to verify. However, it might be possible to remove this assumption by replacing the derivative of the Riemannian energy, in the proof, with its Dini derivative \cite{zhao2022tube, singh2023robust}. 


Theorem \ref{thm:LMI_characterization} is motivated by the recent advancements towards reducing the $L_2$ gain of nonlinear systems by designing feedback control \cite{manchester2018robust}. Indeed, one may regard $u_d$ as an external control input and $w$ as an external disturbance in the composite system such that \eqref{eq:similarity_comp} is analogous to the $H_{\infty}$ control problem.
However and as also stated in \cite{pirastehzad2024comparison}, the problem differs in that it seeks the existence of the constants $\gamma,\delta,\mu,$ and $\eta$ as opposed to be specified a priori.
The LMI in equation \eqref{eq:LMI_W} may depend on $u_d(t)$ and $w(t)$ due to the presence of terms $A$ and $\dot{W}$. Dependence on $w(t)$ is not restrictive as a bound on the $w$ can be pre-established and incorporated in equation \eqref{eq:LMI_W}. Analogous to \cite{zhao2022tube,manchester2017control}, one may assume the following to remove the dependence on $u_d(t)$:
\begin{enumerate}
    \item System $\Sigma_2$ is affine in $u_d(t)$, i.e., of the form 
    \begin{equation*}
        \Sigma_2: \begin{cases}
            \dot{x_2} &= f_2(x_2,u_2,d_2) + E_d(x_2)u_d,\\
            y_2 &= g_2(x_2),
        \end{cases}
    \end{equation*}
    and
    \item for each $i\in\{1,\dots,q_d\}$, $\partial_{e_i}W-\frac{\partial e_i}{\partial x}W - W\left(\frac{\partial e_i}{\partial x}\right)^\top = 0$, where $e_i$ denotes the $i$th column of $E_d(x)$.
\end{enumerate}

Requiring $\Sigma_2$ to be affine in $u_d(t)$ is not restrictive. Indeed, control-affine systems is a broad class of systems and even if there is a nonlinear dependence on $u_d(t)$, in many cases such systems can be transformed into a control-affine system \cite{chen1999model}. Additionally, since
we aim to design $u_d(t)$, requiring $u_d(t)$ to act in an affine manner is a design choice.  
The second condition states that $e_i$ is a Killing vector for the metric $W(x)$, i.e., if $E_d$ is in the form of $\begin{bmatrix}
    \mathbf{0}, \mathbf{I}_{q_d'\times q_d'}
\end{bmatrix}^\top$, then the second condition requires that $W(x)$ must not depend on the last $q_d'$ state components.

\section{Necessary and Sufficient Conditions}\label{sec:conditions}

Although Theorem \ref{thm:LMI_characterization} provides an LMI characterization of the notion of $(T_e,\gamma,\delta)$-similarity, it does not provide insights on when equation \eqref{eq:LMI_W} is feasible. Conditions ensuring feasibility of equation \eqref{eq:LMI_W} are certainly desirable as they provide insights towards the class of systems that are $(T_e,\gamma,\delta)$-similar. To this end, in Theorem \ref{thm:necessary} and Theorem \ref{thm:sufficient}, we will characterize a necessary and a sufficient condition, respectively, on the feasibility of equation \eqref{eq:LMI_W}.

\begin{theorem}\label{thm:necessary}
    Suppose equation \eqref{eq:LMI_W} is feasible. Then, there exists a matrix $W_{11}(x)\succ 0$ that satisfies
    \begin{align}\label{eq:necessary}
        \dot{W}_{11}(x) - A_1W_{11}(x) - W_{11}(x)A_1^\top \geq 0,
    \end{align}
    where $A_1 = \frac{\partial f_1}{\partial x_1}$. 
\end{theorem}
\begin{proof}
    If equation \eqref{eq:LMI_W} holds, then that implies that $\mathcal{W}\geq 0$, where $\mathcal{W} = \dot{W} -WA^\top -AW-EY -Y^\top E^\top$. 
    From equation \eqref{eq:E}, 
    \begin{align*}
        EY = \begin{bmatrix}
            \mathbf{0} & \mathbf{0}\\
            E_dY_1 & E_dY_2
        \end{bmatrix},
    \end{align*} 
    where $Y = \begin{bmatrix}
        Y_1 & Y_2
    \end{bmatrix}$. Let $W(x) = \begin{bmatrix}
        W_{11} & W_{12}\\
        W_{12}^\top & W_{22}
    \end{bmatrix}$.
    Then, the $(1,1)$ entry, i.e., the element in the first row and first column, of $\mathcal{W}$ is $\dot{W}_{11}(x) - A_1W_{11}(x) - W_{11}(x)A_1^\top$. Since $\mathcal{W}\ge 0$, $\dot{W}_{11}(x) - A_1W_{11} - W_{11}A_1^\top\ge 0$ must hold.
\end{proof}

Equation \eqref{eq:necessary} implies that system $\Sigma_1$ must be incrementally stable \cite{forni2013differential} under zero input $u_1$ and disturbance $d_1$. 
Intuitively, equation \eqref{eq:necessary} is satisfied by a class of systems for which the infinitesimal displacement at fixed time between two neighboring trajectories of a system in that class is non-increasing.
Requiring $\Sigma_1$ to satisfy equation \eqref{eq:necessary} is not restrictive as there are many systems (for e.g. contractive systems \cite{davydov2024perspectives}) that satisfy equation \eqref{eq:necessary}.
For LTI systems, equation \eqref{eq:necessary} implies that system $\Sigma_1$ must be $0$-asymptotically stable.

Theorem \ref{thm:necessary} establishes a necessary condition for the feasibility of equation \eqref{eq:LMI_W} and, as a consequence, a necessary condition for $(T_e,\gamma,\delta)$-similarity (under the assumptions required for Theorem \ref{thm:LMI_characterization}). Theorem \ref{thm:necessary} is valuable as it characterizes a property that $\Sigma_1$ must satisfy such that $\Sigma_2$ is $(T_e,\gamma,\delta)$-similar to $\Sigma_1$. However, it does not characterize any condition for system $\Sigma_2$. The next result provides a sufficient condition which depends on both system $\Sigma_1$ and $\Sigma_2$ for the feasibility of equation \eqref{eq:LMI_W}.

\begin{theorem}\label{thm:sufficient}
    Suppose that, for some $\lambda>0$, the composite system  \eqref{eq:dynamics_composite} satisfies 
    \begin{equation}\label{eq:CCM}
    \dot{W} - WA^\top - AW - EY -Y^\top E^\top > 2\lambda W.
\end{equation}
    Then, on any compact subset of $x, w, u_d$, system $\Sigma_2$ is $(T_e,\gamma,\delta)$-similar to system $\Sigma_1$.
\end{theorem}
\begin{proof}
    Since $\lambda>0$ and $W$ is bounded, there exists sufficiently small $\eta$ and sufficiently large $\gamma,\delta, \mu$ such that
    $2\lambda W > \mathcal{C}^\top R\mathcal{C} + BQ^{-1}B^\top$ holds. Using equation \eqref{eq:CCM},
    \begin{align*}
        \dot{W} - WA^\top - AW - EY -Y^\top E^\top > \mathcal{C}^\top R\mathcal{C} + BQ^{-1}B^\top.
    \end{align*}
    The claim then follows upon applying Schur's complement twice and then using Theorem \ref{thm:LMI_characterization}.
\end{proof}

Given that two nonlinear systems characterized in \eqref{eq:dynamics_1} and \eqref{eq:dynamics_2} are similar, in the sense of Definition \ref{def:similarity}, a natural question is whether their differential dynamics are also similar, with the notion of similarity defined analogously to that in Definition \ref{def:similarity}. 
In the next section, we will establish that the answer to this question is in the affirmative. We will also establish that similarity, that we will define formally shortly, of differential dynamics implies similarity of their respective nonlinear systems, in the sense of Definition \ref{def:similarity}.

\section{Properties of \texorpdfstring{$(T_e,\gamma,\delta)$}{(Te,\gamma,\delta)}-Similarity}\label{sec:properties}

In this section, we provide an affirmative answer to the question of whether $(T_e,\gamma,\delta)$-similarity of nonlinear systems imply analogous similarity between their respective differential dynamics.
We begin by introducing a notion of $(T_e,\gamma,\delta)$-similarity between the differential dynamics $\Delta\Sigma_i$ of $\Sigma_i$, for $i\in\{1,2\}$.

\begin{definition}[Differential $(T_e,\gamma,\delta)$-Similarity]\label{def:diff_similarity}
    Given differential dynamics $\Delta\Sigma_1$ and $\Delta\Sigma_2$ of systems $\Sigma_1$ and $\Sigma_2$, respectively, system $\Delta\Sigma_2$ is said to be differentially $(T_e,\gamma,\delta)$-similar to $\Delta\Sigma_1$ if there exist positive constants $\gamma,\delta,\mu,\eta$ such that for every  $\Delta u_1,\Delta u_2,\Delta d_1, \Delta d_2 \in \mathcal{L}_{2,T}$ and every $\Delta x_{i,0}$, $i\in \{1,2\}$, there exist a driving input $\Delta u_d\in \mathcal{L}_{2,T}$ such that, for all finite $T$ satisfying $0<T<T_e$
    \begin{equation}\label{eq:similarity_def_diff}
    \begin{split}
        &\norm{\Delta y_1- \Delta y_2}^2_{[0,T]} \leq \gamma\norm{\Delta u_1-\Delta u_2}^2_{[0,T]}\\
        &+ \delta \norm{\col{\Delta u_1,\Delta u_2}}^2_{[0,T]} + \mu\norm{\col{\Delta d_1,\Delta d_2}}^2_{[0,T]} \\
        &-\eta\norm{\Delta u_d}^2_{[0,T]}+ \Delta b\left(x_{0},\Delta x_{1,0},\Delta x_{2,0}\right),
    \end{split}
    \end{equation}
    for some finite-valued continuous function $\Delta b\left(x_{1,0},x_{2,0},\Delta x_{1,0},\Delta x_{2,0}\right)\ge 0$ and $\Delta b\left(x_{0},\mathbf{0},\mathbf{0}\right) = 0$ for all $x_{1,0},x_{2,0}$.
\end{definition}

Analogously to that for \eqref{eq:similarity_def}, we assume that the initial conditions $\Delta x_0 = \col{\Delta x_{1,0},\Delta x_{2,0}}$, $x_0$ and the signal $\Delta w = \col{\Delta u_1, \Delta u_2, \Delta d_1}$ are such that $\Delta b(x_0,\Delta x_0)=0$ and $\Delta w=0$ do not jointly hold.
Intuitively, the notion of differential $(T_e,\gamma,\delta)$-similarity characterizes the similarity between the variation in the outputs of system $\Sigma_1$ and $\Sigma_2$ by accounting for the variation in the inputs $u_1(t)$, $u_2(t)$, and the variation in the disturbances $d_1(t), d_2(t)$. Formally, the notion of differential $(T_e,\gamma,\delta)$-similarity characterizes the similarity of $\Delta\Sigma_1$ and $\Delta\Sigma_2$ in terms of variations in the input-output behavior of $\Sigma_1$ and $\Sigma_2$. The constant $\gamma$ characterizes the effect of variation on the dissimilar inputs and the constant $\delta$ characterizes the effect of individual input variation. The constant $\mu$ characterizes the effect of variation on the disturbances $d_1$ and $d_2$.

We begin by establishing that $(T_e,\gamma,\delta)$-similarity can be completely characterized by dissipativity of the differential composite system $\Delta \Sigma$ defined in equation \eqref{eq:diff_dyn_composite}.

\begin{lemma}\label{lem:diss_characterization}
    Suppose Assumptions \ref{assum:d_2_feedback} and \ref{assum:zero_measure} hold. Then, system $\Sigma_2$ is $(T_e,\gamma,\delta)$-similar to system $\Sigma_1$ if and only if system $\Delta\Sigma$ is dissipative with respect to the supply rate 
    \begin{align}\label{eq:supply_rate}
        s_{\Delta}\left(\Delta w(t),\Delta z(t)\right) = \begin{bmatrix}
            \Delta w\\
            \Delta z
        \end{bmatrix}^{\top}\begin{bmatrix}
            Q & \mathbf{0}\\
            \mathbf{0} & -R
        \end{bmatrix}\begin{bmatrix}
            \Delta w\\
            \Delta z
        \end{bmatrix},
    \end{align}
    i.e., there exists a function $V_{\Delta}: \mathbb{R}^n \times \mathbb{R}^n \to [0,\infty)$ such that, for all $t_1\leq t_2$
    \begin{multline}\label{eq:diff_diss}
        V_{\Delta}\left(x(t_2),\Delta x(t_2)\right)-V_{\Delta}\left(x(t_1),\Delta x(t_1)\right)\\
        \leq \int_{t_1}^{t_2}s_{\Delta}\left(\Delta w(t),\Delta z(t)\right)dt.
    \end{multline}
    Moreover, $V_{\Delta} = \Delta x^{\top} M(x)\Delta x$, where $M(x)$ satisfies equation \eqref{eq:LMI_M}.
\end{lemma}
\begin{proof}
    We begin by establishing that if system $\Sigma_2$ is $(T_e,\gamma,\delta)$-similar to system $\Sigma_1$, then $\Delta \Sigma$ is dissipative.
    Since system $\Sigma_2$ is $(T_e,\gamma,\delta)$-similar to system $\Sigma_1$, equation \eqref{eq:LMI_M} holds. Pre- and post-multiplying \eqref{eq:LMI_M} with $\col{\Delta w,\Delta z}^\top$ yields
    \begin{multline*}
        \Delta x^\top \dot{M}\Delta x + 2\Delta xM(A\Delta x+B\Delta w) \\
        \leq \Delta w^\top Q\Delta w - \Delta x^\top C^\top RC\Delta x\\
        = \frac{d}{dt}\left(\Delta x^\top M\Delta x\right) \leq \begin{bmatrix}
            \Delta w & \Delta z
        \end{bmatrix}\begin{bmatrix}
            Q & 0\\
            0 & -R
        \end{bmatrix}\begin{bmatrix}
            \Delta w \\ \Delta z
        \end{bmatrix}.
    \end{multline*}
    Integrating with respect to time from $t_1$ to $t_2$ yields equation \eqref{eq:diff_diss}.

    We now show that if $\Delta \Sigma$ is dissipative then, system $\Sigma_2$ is $(T_e,\gamma,\delta)$-similar to $\Sigma_1$.
    Differentiating equation \eqref{eq:diff_diss} with respect to time yields
    \begin{align*}
        \begin{bmatrix}
            \Delta w & \Delta z
        \end{bmatrix}\begin{bmatrix}
            \dot{M} + M\mathcal{A} + \mathcal{A}^\top M + \mathcal{C}^\top R \mathcal{C} & MB \\
            B^\top M & -Q
        \end{bmatrix}\begin{bmatrix}
            \Delta w \\ \Delta z
        \end{bmatrix} \leq 0.
    \end{align*}
    This implies that equation \eqref{eq:LMI_M} holds. The claim then follows from Theorem \ref{thm:LMI_characterization} and the fact that equation \eqref{eq:LMI_M} and equation \eqref{eq:LMI_W} are equivalent.
\end{proof}

From Lemma \ref{lem:diss_characterization}, the notion of $(T_e,\gamma,\delta)$-similarity for nonlinear systems is completely characterized by the notion of differential dissipativity of the composite system. In other words, if system $\Sigma_1$ and the system $\Sigma_2$ are such that $\Delta\Sigma$ is dissipative, then it follows from Lemma \ref{lem:diss_characterization}, system $\Sigma_2$ is $(T_e,\gamma,\delta)$-similar to $\Sigma_1$. Alternatively, if system $\Sigma_2$ is $(T_e,\gamma,\delta)$-similar to $\Sigma_1$, then from Lemma \ref{lem:diss_characterization}, it follows that $\Delta\Sigma$ is dissipative.

\begin{theorem}\label{thm:sim_implies_diff_sim}
    Suppose Assumptions \ref{assum:d_2_feedback} and \ref{assum:zero_measure} hold and system $\Sigma_2$ is $(T_e,\gamma,\delta)$-similar to system $\Sigma_1$. Then, system $\Delta \Sigma_2$ is differentially $(T_e,\gamma,\delta)$-similar to system $\Delta \Sigma_1$.
\end{theorem}
\begin{proof}
    From Lemma \ref{lem:diss_characterization}, system $\Delta\Sigma$ is dissipative with respect to supply rate \eqref{eq:supply_rate}. For the choice of $t_1=0$ and $t_2=T$, it follows that $\norm{\Delta z}_{R,T}^2\leq \norm{\Delta w}_{Q,T}^2 + \Delta b$ which is equivalent to equation \eqref{eq:similarity_def_diff}.
\end{proof}

Theorem \ref{thm:sim_implies_diff_sim} establishes that if system $\Sigma_2$ is similar to $\Sigma_1$ in the sense of Definition \ref{def:similarity}, then $\Delta\Sigma_2$ is similar to $\Delta\Sigma_1$ in the sense of Definition \ref{def:diff_similarity}. 
This raises the question whether differential $(T_e,\gamma,\delta)$-similarity imply $(T_e,\gamma,\delta)$-similarity of the underlying nonlinear systems.
To address this question, the following result relates differential $(T_e,\gamma,\delta)$-similarity with dissipativity of system $\Delta\Sigma$ defined in \eqref{eq:diff_dyn_composite}.

\begin{lemma}\label{lem:diff_dyn_diss}
    Suppose Assumptions \ref{assum:d_2_feedback} and \ref{assum:zero_measure} hold. Further, suppose that system $\Delta \Sigma_2$ is differentially $(T_e,\gamma,\delta)$-similar to system $\Delta\Sigma_1$ and let $Q$ and $R$ be defined in \eqref{eq:Q_and_R}. Then, system $\Delta \Sigma$ defined in \eqref{eq:diff_dyn_cl} is dissipative with respect to the supply rate \eqref{eq:supply_rate}.
\end{lemma}
\begin{proof} 
    Since $\Delta \Sigma_2$ is differentially $(T_e,\gamma,\delta)$-similar to system $\Delta \Sigma_1$, it follows that $-\norm{\Delta w}_{Q,T}^2 + \norm{\Delta z}_{R,T}^2 \leq \Delta b$ for all $\Delta w$ and $\Delta x_0$. This means that for all $\Delta w \in \mathcal{L}_{2,T}$ and $\Delta x_0$, $-s_{\Delta}\left(\Delta w(t),\Delta z(t)\right)<\infty$, i.e., is finite.
    Thus, from \cite[Theorem 3.1.11]{van2000l2}, the claim follows.
\end{proof}

\begin{remark}
    It is well known that differential dissipativity implies stability of system $\Delta \subscr{\Sigma}{cl}$ (cf. \eqref{eq:diff_dyn_cl}) when $\forall \Delta z$, $s_\Delta(0,\Delta z)\leq 0$ \cite{verhoek2023convex}. In this work, this holds naturally since $-R\prec 0$ (as $\eta>0$). 
\end{remark}

\begin{remark}
    Following analogous steps as in Lemma \ref{lem:diff_dyn_diss}, it can be established that if system $\Sigma_2$ is $(T_e,\gamma,\delta)$-similar to system $\Sigma_1$, then system $\Sigma$ is dissipative with analogously defined supply rate.
\end{remark}



\begin{theorem}\label{thm:iff}
    Suppose Assumptions \ref{assum:d_2_feedback} and \ref{assum:zero_measure} hold hold. Further, suppose that for an arbitrary choice of $u_d(t)$, there exists a $\Delta u_d(t)$ such that system $\Delta\Sigma_2$ is $(T_e,\gamma,\delta)$-similar to system $\Delta\Sigma_1$. Then, system $\Sigma_2$ is $(T_e,\gamma,\delta)$-similar to system $\Sigma_1$ for the choice of $u_d(t)$. 
\end{theorem}
\begin{proof}
    The proof directly follows from Lemma \ref{lem:diff_dyn_diss} and Lemma \ref{lem:diss_characterization}.
\end{proof}

Theorem \ref{thm:iff} implies that characterizing similarity of differential dynamics of \eqref{eq:diff_dyn_cl}, for an arbitrary $d_2$, yields similarity of the nonlinear system with the choice of the specified $u_d$. Since the differential dynamics $\Delta \Sigma_1$ and $\Delta \Sigma_2$ of $\Sigma_1$ and $\Sigma_2$, respectively, represent a Linear Time-Varying system (LTV) (cf. \eqref{eq:diff_dyn_1} and \eqref{eq:diff_dyn_2}, respectively), the problem of determining whether the given nonlinear systems are $(T_e,\gamma,\delta)$-similar might be easier to solve by determining whether the associated differential dynamics are similar. 

Next, we establish that $(T_e,\gamma,\delta)$-similarity satisfies reflexivity and transitivity properties, further motivating $(T_e,\gamma,\delta)$-similarity as a notion of system comparison.
We begin with recalling the definition of minimal state-space realization of $\Delta\Sigma$.

\begin{definition}
    System $\Delta\Sigma$ is said to have a minimal realization if the pair $[A(t), B(t)]$ is uniformly controllable and the pair $[A(t), C(t)]$ is uniformly observable.
\end{definition}

\begin{proposition}\label{prop:reflexivity}
    Consider system $\Sigma$ of the following form:
    \begin{equation}\label{eq:similarity_itself_1}
       \Sigma: \begin{cases}
        \dot{x}_1 = f(x_1,u_1,d_1), \quad x_1(0) = x_{1,0};\\
        y_1 = g(x_1).
    \end{cases} 
    \end{equation}
    Now consider the same system, denoted as $\Sigma_d$, but with input $u_2$, disturbance $d_2$, and a perturbation input $u_d$:
    \begin{equation}\label{eq:similarity_itself_2}
       \Sigma_d: \begin{cases}
        \dot{x}_2 = f(x_2,u_2,d_2+u_d), \quad x_2(0) = x_{1,0};\\
        y_2 = g(x_2).
    \end{cases} 
    \end{equation}
    Suppose, under zero disturbance, the state-space realization of the differential dynamics $\Delta\Sigma$ of system $\Sigma$ is minimal and bounded. Further suppose that, under zero input and zero disturbance, $\Delta \Sigma$ is exponentially stable.
    Then, system $\Sigma_d$ is $(T_e,\gamma,\delta)$-similar to system $\Sigma$, or equivalently $\Sigma$ is $(T_e,\gamma,\delta)$-similar to itself. 
\end{proposition}
\begin{proof}
    Select $u_d =d_1-d_2$ and define $e_x = x_1-x_2$ and $e_y = y_1-y_2$. The system corresponding to the state $e_x$ and output $e_y$ is:
    \begin{equation}
       \Sigma_e: \begin{cases}
        \dot{e}_x = f(x_1,u_1,d_1) - f(x_2,u_2,d_1); e_x(0) = 0\\
        e_y = g(x_1) - g(x_2).
    \end{cases} 
    \end{equation}
    Note that if $u_1=u_2$, then $e_x=e_y=0$ as $x_1(0) = x_2(0)$.
    The differential dynamics of system $\Sigma_e$, denoted as $\Delta\Sigma_e$, with state $\Delta e_x$ and output $\Delta e_y$ is: 
    \begin{equation}
    \Delta \Sigma_e: 
        \begin{cases}
            \Delta \dot{e}_x = A e_x + B\left(\Delta u_1 - \Delta u_2\right); \Delta e_x(0) = 0\\
            \Delta e_y = Ce_x.
        \end{cases}
    \end{equation}
    Since the differential dynamics of $\Sigma$ are assumed to be exponentially stable, system $\Delta \Sigma_e$, under zero input, is exponentially stable. It follows from \cite[Theorem 3]{silverman1968controllability} (see also \cite[Lemma 3.6]{fromion2003theoretical}) that there exist a finite constant $\gamma>0$ such that
    \begin{align*}
        \norm{\Delta e_y}_{[0,T]}^2\leq \gamma \norm{\Delta u_1-\Delta u_2}_{[0,T]}^2.
    \end{align*}
    Denote $\tilde{f}(x_1,x_2,u) = f(x_1,u_1,d_1) - f(x_2,u_2,d_1)$ and $\tilde{g}(x_1,x_2) = g(x_1)-g(x_2)$, where $u\coloneqq \col{u_1,u_2,d_1}$.
    From \cite[Theorem 3.4]{fromion2003theoretical}, it follows that for any inputs $u\in \mathcal{L}_{2,T}$ and $u'=\col{u_1',u_2',d_1'}\in \mathcal{L}_{2,T}$, system $\Sigma_e$ is incrementally stable, i.e.,
    \begin{multline*}
        \norm{y_1-y_2-y_1'+y_2'}_{[0,T]}^2 \leq \\
        \gamma\norm{\col{u_1-u_1',u_2-u_2',d_1-d_1'}}_{[0,T]}^2,
    \end{multline*}
    where $y_1'-y_2'$ is the output of system $\Sigma_e$ under input $u'$.
    Selecting $u_2' = u_1'=u_1$ and since $x_1(0) = x_2(0)$ yields $y_1'-y_2' = 0$. Further, select $d_1' = 2d_1$. Then, it follows that for any $u_1,u_2,d_1$,
    \begin{align*}
        \norm{y_1-y_2}_{[0,T]}^2 &\leq \gamma\norm{\col{u_2-u_1,-d_1}}_{[0,T]}^2\\
        &= \gamma\norm{u_1-u_2}_{[0,T]}^2 + \gamma\norm{d_1}_{[0,T]}^2.
    \end{align*}
    For any $\delta>0$, $\eta<2\gamma$, and $\mu=\gamma$,
    \begin{multline*}
        \norm{y_1-y_2}_{[0,T]}^2 \leq \gamma\norm{u_1-u_2}_{[0,T]}^2 + \delta\norm{\col{u_1,u_2}}_{[0,T]}^2\\
        +\mu\norm{\col{d_1,d_2}}_{[0,T]}^2-\eta\norm{d_1-d_2}_{[0,T]}^2.
    \end{multline*}
    The claim follows since $b(x_1(0),x_1(0))=0$ for any $x_1(0)$.
\end{proof}

The next result establishes that the $(T_e,\gamma,\delta)$-similarity also exhibits the transitivity property.

\begin{proposition}\label{prop:transivity}
    Given systems $\Sigma_1$, $\Sigma_2$, and $\Sigma_3$, let $T_e = \min\{T_e^1,T_e^2,T_e^3\}$.
    Suppose system $\Sigma_2$ is $\left(T_e,\gamma_{12},\delta_{12}\right)$-similar to system $\Sigma_1$ and system $\Sigma_3$ is $(T_e,\gamma_{23},\delta_{23})$-similar to system $\Sigma_2$. Then, system $\Sigma_3$ is $(T_e,\gamma_{13},\delta_{13})$-similar to system $\Sigma_1$ if $b_{13}\left(x_{1,0},x_{3,0}\right)\geq b_{12}\left(x_{1,0},x_{2,0}\right) + b_{23}\left(x_{2,0},x_{3,0}\right)$.
\end{proposition}
\begin{proof}
    Since $\Sigma_2$ is $\left(T_e,\gamma_{12},\delta_{12}\right)$-similar to $\Sigma_1$
    \begin{multline}\label{eq:trans_12}
        \norm{y_1-y_2}^2_{[0,T]} \leq \gamma_{12}\norm{u_1-u_2}^2_{[0,T]} + \delta_{12} \norm{\col{u_1,u_2}}^2_{[0,T]} \\
        + \mu_{12}\norm{\col{d_1,d_2}}^2_{[0,T]} -\eta_{12}\norm{u_d^2}^2_{[0,T]} \\
        + b_{12}\left(x_{1,0},x_{2,0}\right),
    \end{multline}
    where $u_d^2$ denotes the perturbation input applied to $\Sigma_2$.
    Similarly, since $\Sigma_3$ is $\left(T_e,\gamma_{23},\delta_{23}\right)$-similar to $\Sigma_2$
    \begin{multline}\label{eq:trans_23}
        \norm{y_2-y_3}^2_{[0,T]} \leq \gamma_{23}\norm{u_1-u_2}^2_{[0,T]} + \delta_{23} \norm{\col{u_1,u_2}}^2_{[0,T]} \\
        + \mu_{23}\norm{\col{d_1,d_2}}^2_{[0,T]} -\eta_{23}\norm{u_d^3}^2_{[0,T]} \\
        + b_{23}\left(x_{2,0},x_{3,0}\right),
    \end{multline}
    where $u_d^3$ denotes the perturbation input applied to $\Sigma_3$.
    Selecting $u_2 = \frac{1}{2}\left(u_1+u_3\right)$ and $d_2 = 0$ followed by adding \eqref{eq:trans_12} and \eqref{eq:trans_23}
    \begin{multline*}
        \frac{1}{2}\norm{y_1-y_3}^2_{[0,T]} \leq \frac{\left(\gamma_{12}+\gamma_{23}\right)}{4}\norm{u_1-u_3}^2_{[0,T]} +\\
        \left(\delta_{12}+\delta_{23}\right)\left(\frac{3}{2}\norm{u_1}^2_{[0,T]} + \frac{1}{2}\norm{u_3}^2_{[0,T]} \right)+ 
        \mu_{12}\norm{d_1}^2_{[0,T]} \\
        + \mu_{23}\norm{d_3}^2_{[0,T]} 
        -\eta_{12}\norm{u_d^2}^2_{[0,T]} -\eta_{23}\norm{u_d^3}^2_{[0,T]}\\
        + b\left(x_{1,0},x_{2,0}\right) + b\left(x_{2,0},x_{3,0}\right).
    \end{multline*}

    Selecting $\gamma_{13}=\frac{\gamma_{12}+\gamma_23}{2}$, $\delta_{13}=\left(\delta_{12}+\delta_{23}\right)$, $\mu_{13} = 2\max\{\mu_{12},\mu_{23}\}$, and $\eta_{13} = 2\eta_{23}$ followed by using that $b_{13}\left(x_{1,0},x_{3,0}\right)\geq b_{12}\left(x_{1,0},x_{2,0}\right) + b_{23}\left(x_{2,0},x_{3,0}\right)$  yields
    \begin{multline*}
        \norm{y_1-y_3}^2_{[0,T]} \leq \gamma_{13}\norm{u_1-u_3}^2_{[0,T]} +\delta_{13}\norm{\col{u_1,u_3}}_{[0,T]}^2 \\
        + \mu_{13}\norm{\col{d_1,d_3}}_{[0,T]}^2-\eta_{13}\norm{u_d^3}^2_{[0,T]}+b_{13}\left(x_{1,0},x_{3,0}\right)
    \end{multline*}
    and the claim is established.
\end{proof}

Proposition \ref{prop:transivity} imposes some mild restrictions on the initial conditions of the systems.
In other words, Proposition \ref{prop:transivity} requires $b_{13}\left(x_{1,0},x_{3,0}\right)\geq b_{12}\left(x_{1,0},x_{2,0}\right) + b_{23}\left(x_{2,0},x_{3,0}\right)$.
There are cases for which this condition trivially holds such as when the three systems have identical initial conditions.
Additionally, observe that the proof of Proposition \ref{prop:transivity} does not use the fact that the systems have nonlinear dynamics. Thus, we conjecture that the transitivity property of $(T_e,\gamma,\delta)$-similarity may hold for a much broader class of systems including hybrid systems.

\section{Applications}\label{sec:numerics}

We now illustrate two applications of $(T_e,\gamma,\delta)$-similarity.
In Section \ref{subsec:App_1}, we will illustrate the application of $(T_e,\gamma,\delta)$-similarity for robust hierarchical control via abstract models.
In Section \ref{subsec:App_2}, we will illustrate the application of $(T_e,\gamma,\delta)$-similarity for \emph{enhancing} available abstract model (such as linearized model or a reduced model). 
While not illustrated in this work, one may also find application of $(T_e,\gamma,\delta)$-similarity in specification verification or component replacement \cite{pirastehzad2024comparison}.
Unless otherwise stated, all of our numerical results have been obtained by using MATLAB 2023b via YALMIP \cite{lofberg2004yalmip} and Mosek.

\subsection{Robust Hierarchical Control of Planar Vertical and Take Off Landing (PVTOL) Aircraft}\label{subsec:App_1}

The idea is to use an abstract model, denoted as $\Sigma_1$ from now on in this subsection, for control design. Once the control input via $\Sigma_1$ is determined, a suitable control input determined via the abstraction can then be applied to  the PVTOL $\Sigma_2$.
We will see shortly that if $\Sigma_2$ is $(T_e,\gamma,\delta)$-similar to its abstraction $\Sigma_1$ then, by using this approach, the trajectory of $\Sigma_2$ approximately follows that of $\Sigma_1$, even under the presence of $d_2$, i.e., rejecting the disturbance $d_2$ that acts on $\Sigma_2$.

Consider a PVTOL aircraft \cite{zhao2022tube,cao2021estimating}  with mass $m$, gravitational acceleration $g$, and rotational inertia $J$ with the state vector defined as $x_2 = \begin{bmatrix}
    p_x^Q & p_z^Q & v_x^Q & v_z^Q & \theta & \dot{\theta}
\end{bmatrix}^\top$, where $\begin{bmatrix}
    p_x^Q & p_z^Q
\end{bmatrix}^\top$ denote the position in $x$ and $z$ directions, respectively, $v_x^Q$ and $v_z^Q$ denote the lateral velocity and the velocity along the thrust axis in the body frame, respectively, and $\theta$ denotes the roll angle. The input vector $u_2 = \begin{bmatrix}
    u_{2,1} & u_{2,2}
\end{bmatrix}$ contains the thrust (offset by the weight $mg$ of the vehicle) and the net torque. The dynamics of the vehicle, denoted as $\Sigma_2$ in this section, in the world frame are given by:

\begin{align*}
    \dot{x}_2 &= \begin{bmatrix}
        v_x \\
        v_z \\
        -g\sin(\theta)\\
        g\cos(\theta) - g \\
        \dot{\theta} \\
        0
    \end{bmatrix} +
    \begin{bmatrix}
        0 & 0\\
        0 & 0\\
        -\frac{\sin(\theta)}{m} & 0\\
        \frac{\cos(\theta)}{m} & 0\\
        0 & 0\\
        0 & \frac{1}{J}
    \end{bmatrix}u + 
    \begin{bmatrix}
        0\\
        0\\
        1\\
        0\\
        0\\
        0
    \end{bmatrix}d_2\\
    & + \begin{bmatrix}
        0 & 0\\
        0 & 0\\
        -\frac{0.25}{m} & 0\\
        \frac{0.968}{m} & 0\\
        0 & 0\\
        0 & \frac{1}{J}
    \end{bmatrix}u_d.
\end{align*}

The parameter values were set to $m=0.486$kg and $J=0.00383$ kg m\textsuperscript{2}. The signal $d_2$ was set to  $d_2=0.8+0.2\sin\left(\tfrac{2\pi t}{10}\right)$ to represent wind disturbance affecting $\Sigma_2$. Note that this choice of disturbance is only for simulation purposes; the $(T_e,\gamma,\delta)$-similarity framework in this work holds for any unknown disturbance.
The system matrix corresponding to $u_d$ is chosen to be a constant so that the differential dynamics are independent of $u_d$. Alternatively, one may use dynamics in the body frame instead of world frame. 

As is common in many robotic examples, we use a 2-D double integrator model as the simplified model (i.e. abstraction $\Sigma_1$) of $\Sigma_2$. 
Formally, we consider a linear system with states $x_1 = \begin{bmatrix}
    p_x^D & p_z^D & v_x^D & v_z^D
\end{bmatrix}$, where $\begin{bmatrix}
    p_x^D & p_z^D
\end{bmatrix}^\top$ denote the position in $x$ and $z$ directions, respectively, and $v_x^D$ and $v_y^D$ denote the lateral velocity and the velocity along the thrust axis, respectively. The control input $u_1 = \begin{bmatrix}
    u_{1,1} & u_{1,2}
\end{bmatrix}$ denotes the accelerations in the $x$ and $z$ direction, respectively. 
The dynamics are:
\begin{align*}
    \dot{x}_1 = \underbrace{\begin{bmatrix}
        0 & 0 & 1 & 0 \\
        0 & 0 & 0 & 1 \\
        0 & 0 & 0 & 0 \\
        0 & 0 & 0 & 0
    \end{bmatrix}}_{A_{\text{ol}}} x_1 +
     \underbrace{\begin{bmatrix}
        0 & 0 \\
        0 & 0 \\
        1 & 0 \\
        0 & 1
    \end{bmatrix}}_{B_{\text{ol}}}u_1.
\end{align*}
We set $d_1=0$ as the aim is that the output of $\Sigma_2$ follows that of $\Sigma_1$ while rejecting the disturbance affecting $\Sigma_2$. 
To ensure that the necessary condition for $(T_e,\gamma,\delta)$-similarity holds (cf. Theorem \ref{thm:necessary}), we design a gain matrix $K_{\Sigma_1}$ such that the closed-loop dynamics of $\Sigma_1$ is stable. Specifically, we design a Linear Quadratic Regulator (LQR) controller for $\Sigma_1$ to yield the closed-loop dynamics
\begin{align}\label{eq:cl_double_integrator}
    \dot{x}_1 = \left(A_{\text{ol}} + B_{\text{ol}}K_{\Sigma_1}\right)x_1,
\end{align}
where 
\begin{align*}
    A_{\text{ol}} + B_{\text{ol}}K_{\Sigma_1} = A_1 = \begin{bmatrix}
        0 & 0 & 1 & 0 \\
        0 & 0 & 0 & 1 \\
        -1 & 0 & -1.732 & 0 \\
        0 & -1 & 0 & -1.732
    \end{bmatrix}.
\end{align*}

To determine whether $\Sigma_2$ is $(T_e,\gamma,\delta)$-similar to closed-loop $\Sigma_1$ (cf. \eqref{eq:cl_double_integrator}), or equivalently, to solve equation \eqref{eq:LMI_W}, we impose the following bounds: $\theta \in [-\frac{\pi}{3}, \frac{\pi}{3}]$, $u_{2,1}\in [-3, 3]$.
To solve equation \eqref{eq:LMI_W}, we first discretized the space of $\theta$ and $u_{2,1}$ to construct an LMI (cf. \eqref{eq:LMI_W}) for each discretized value. 
We then searched for a quadratic matrix function $Y(\theta)$ and a constant matrix $W$ with $W \succeq  10^{-6}\mathbf{I}$ that satisfies every LMI constructed in the discretized space using MATLAB 2023b and YALMIP on a PC with 8 GB RAM. 
We observed that fixing $\gamma=0.001$, $\delta=0.01$, and $\mu=0.01$ leads to an improved performance by the solver.
The total CPU time to obtain $W$ and $Y(\theta)$, as reported by YALMIP, was approximately $3.2$ seconds. 
Upon obtaining $W$ and $Y(x)$, we determine $u_d$ via equation \eqref{eq:d_2_form}. 
Since the abstract system and the PVTOL system do not have the same number of states, to compute $x^*(t)$, we used a $\begin{bmatrix}
    \mathbf{I}_{4\times 4} & \mathbf{0}_{4\times 2}^\top
\end{bmatrix}$ to project the state of $\Sigma_1$ to $\Sigma_2$.
We note that advanced optimization techniques such as Sum-of-Squares (SOS) programming \cite{papachristodoulou2005tutorial} may be used to solve equation \eqref{eq:LMI_W}. However, for large-scale systems, SOS is known to suffer from scalability issues \cite{ahmadi2017improving}.

Since $(T_e,\gamma,\delta)$-similarity characterizes the bound for every $u_1$ and $u_2$, we set $u_2=0$ circumventing designing a reference trajectory for the PVTOL and designing a controller for the PVTOL to track the reference trajectory. We present the numerical results in Figure \ref{fig:quad+DI}.

\begin{figure}[t]
    \centering
    \begin{subfigure}[t]{0.45\columnwidth}
        \centering
        \includegraphics[scale = 0.1]{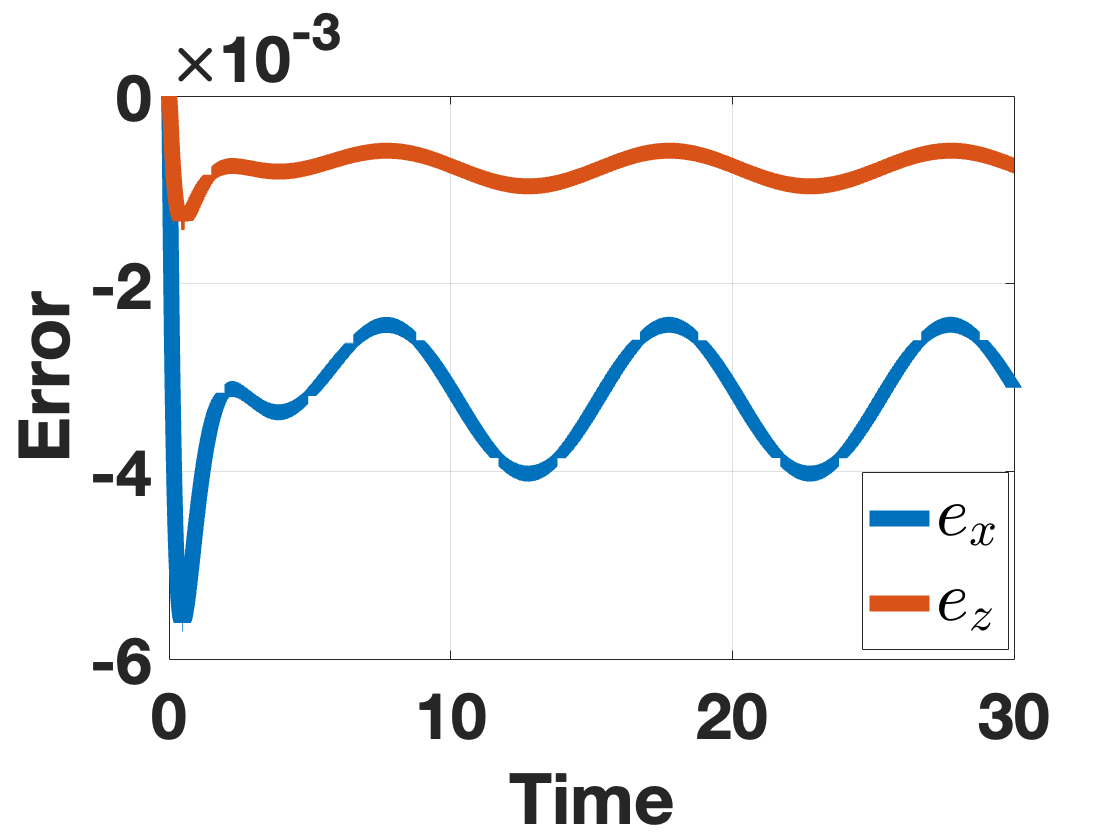}
        \caption{Error between the position of the PVTOL and the double-integrator. The PVTOL is $(T_e,\gamma,\delta)$-similar to the double-integrator.
        The red (resp. blue) curve denotes the error between the position in the $z$ (resp. $x$) direction.}
        \label{fig:fig:quad+DI_error}
    \end{subfigure}%
    \hfill
    \begin{subfigure}[t]{0.45\columnwidth}
        \centering
        \includegraphics[scale = 0.1]{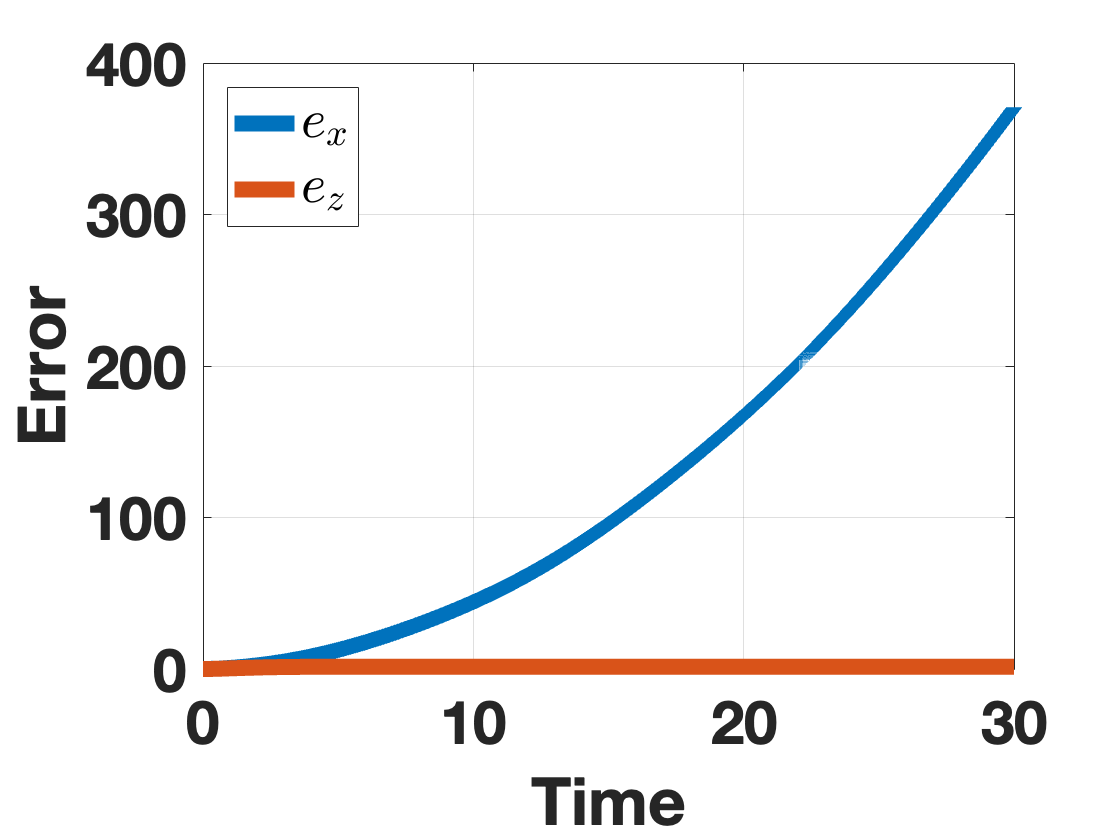}
        \caption{Error between the position of the PVTOL and the double-integrator. The red (resp. blue) curve denotes the error between the position in the $z$ (resp. $x$) direction.}
        \label{fig:fig:quad+DI_no_similarity}
    \end{subfigure}
    \caption{Numerical results for robust hierarchical control of PVTOL using a double-integrator via $(T_e,\gamma,\delta)$-similarity. The disturbance acting on the PVTOL is $0.8+0.2\sin(\tfrac{2\pi t}{10})$.
    (Left) The error between the outputs of the FOM and its ROM when FOM is $(T_e,\gamma,\delta)$-similar to the ROM. (Right) The error between the outputs of the FOM and the ROM without $(T_e,\gamma,\delta)$-similarity.}
    \label{fig:quad+DI}
\end{figure}

From Figure \ref{fig:fig:quad+DI_error}
error between the outputs, i.e., the position, of the closed-loop double integrator and the PVTOL in Figure \ref{fig:fig:quad+DI_error}. From Figure \ref{fig:fig:quad+DI_error}, the error between the position between the PVTOL and its abstraction, i.e., the double-integrator is small meaning that the outputs of the PVTOL approximately follows that of the double-integrator even under the presence of unknown disturbance. In contrast, from Figure \ref{fig:fig:quad+DI_no_similarity}, the error between the $x$-position of the PVTOL and the double-integrator increases with time when the PVTOL is not $(T_e,\gamma,\delta)$-similar to its abstract model.

\subsection{Abstract Model Enhancement and Model Reduction}\label{subsec:App_2}

Another possible application of $(T_e,\gamma,\delta)$-similarity framework presented in this work is to enhance available abstract models, such as linearized or reduced models, of a nonlinear system. 
By enhancing an available abstract model, we mean the following.
\begin{enumerate}
    \item \textbf{Case 1 ($d_1(t)=\mathbf{0}$):} When the nonlinear model is not affected by any disturbance, then by enhancement we mean that we can improve the abstract model's performance such that the output of the abstract model is similar to that of the associated nonlinear model either beyond the regime where linearized model works (when abstract model is a linearized model) or by reducing the error between the outputs of the nonlinear and its abstract model.

    \item \textbf{Case 2 ($d_1(t)\neq \mathbf{0}$):} In applications such as designing digital-twins of reduced complexity of nonlinear system, we require the abstract model to behave similar to the associated nonlinear system even under the presence of arbitrary unknown disturbance $d_1(t)$. 
    Note that, since the nonlinear system is operating in a real environment (and not in a simulation) and the digital-twin operates in parallel at a different location, one cannot apply $d_1(t)$ to the digital-twin as $d_1(t)$ is unknown. In this context, by enhancing the abstract model, we mean to improve the performance of the abstract model such that the output of the abstract model is similar to that of the nonlinear system even under the presence of unknown disturbance $d_1(t)$. In \cite{bajaj2025online}, it has been shown that many reduced models obtained by using classical methods do not satisfy this property.
\end{enumerate}

To utilize $(T_e,\gamma,\delta)$-similarity framework for enhancing abstract models in both case 1 and case 2, we consider the nonlinear system as $\Sigma_1$ and its abstract model as $\Sigma_2$ and determine whether $\Sigma_2$ is $(T_e,\gamma,\delta)$-similar to $\Sigma_1$. If the answer is in the affirmative, we refer to the abstract model along with the application of $d_2(t)$ as $(T_e,\gamma,\delta)$-abstract model, denoted as $(T_e,\gamma,\delta)$-$\Sigma_2$, of $\Sigma_1$. 

We present a numerical study, each for Case 1 and Case 2 and present the results in Figure \ref{fig:enhanced_linearized} and \ref{fig:enhanced_rom}, respectively. For both cases, the implementation details are analogous to that in Section \ref{subsec:App_1} and thus, has been omitted for brevity.

\begin{figure}[t]
    \centering
    \begin{subfigure}[t]{0.45\columnwidth}
        \centering
        \includegraphics[scale = 0.1]{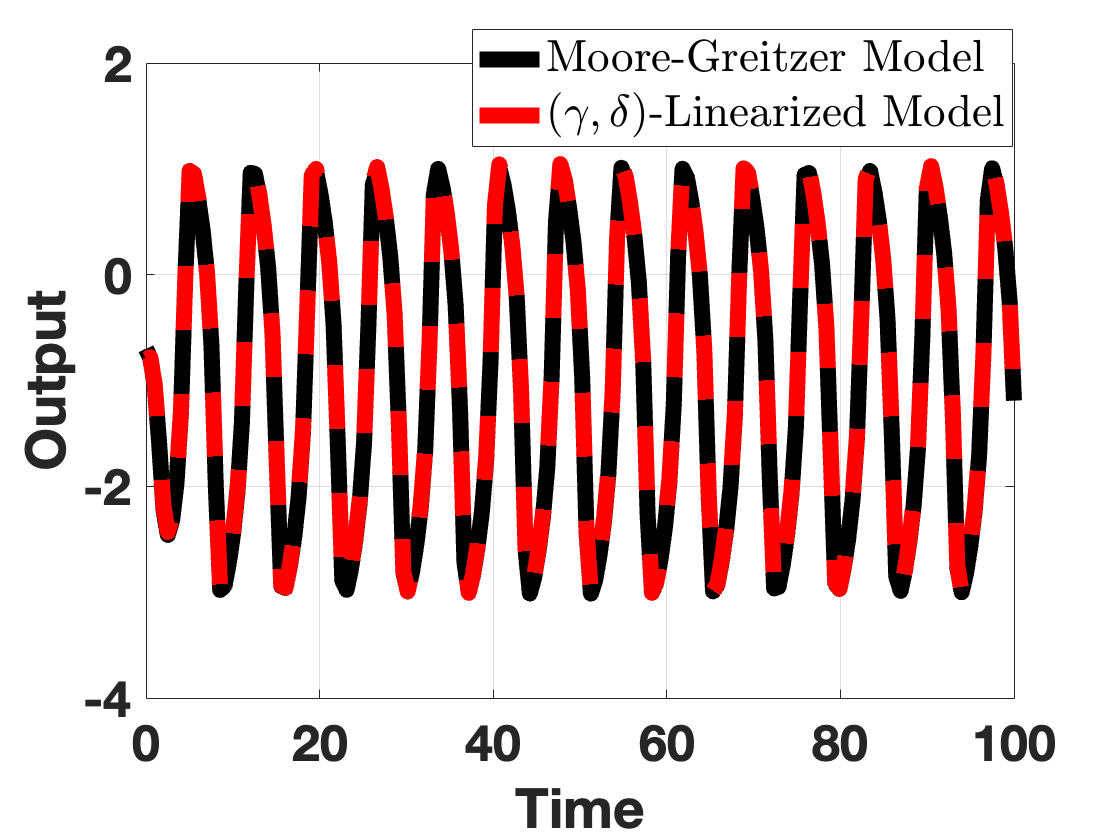}
        \caption{Comparison of outputs of $\Sigma_1$ and $(T_e,\gamma,\delta)$-$\Sigma_2$, respectively.}
        \label{fig:linearized_output}
    \end{subfigure}%
    \hfill
    \begin{subfigure}[t]{0.45\columnwidth}
        \centering
        \includegraphics[scale = 0.1]{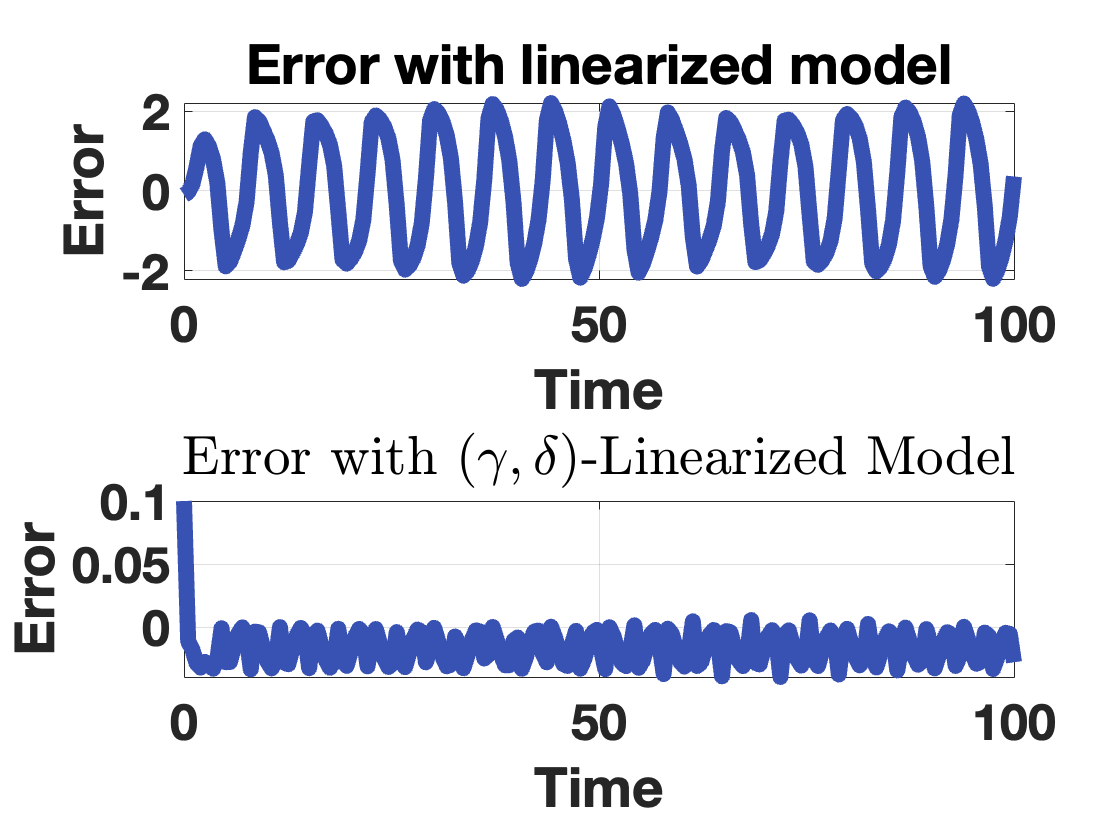}
        \caption{Error between the output of the Moore-Greitzer model and its linearized (top) and its $(T_e,\gamma,\delta)$-linearized model (bottom).}
        \label{fig:linearized_error}
    \end{subfigure}
    \caption{Numerical results for enhancing a linearized model when $d_1(t) = \mathbf{0}$ (Case 1).}
    \label{fig:enhanced_linearized}
\end{figure}

For the numerical study for Case 1, we consider the Moore-Greitzer model as $\Sigma_1$ and its linearized model as $\Sigma_2$ and present the numerical results in Figure \ref{fig:enhanced_linearized}.

Moore-Greitzer model is a simplified model of surge-stall dynamics of a jet engine and is one of the classical models that has motivated substantial development
in nonlinear control design \cite{krstic1995nonlinear,moore1986theory,manchester2018robust}. In particular, the nonlinear model $\Sigma_1$ used in this work is
\begin{align*}
    \begin{bmatrix}
        \dot{\psi}\\
        \dot{\phi}
    \end{bmatrix} = 
    \begin{bmatrix}
        \phi + u_1\\
        -\psi - 1.5\phi^2-0.5\phi^3 + d_1
    \end{bmatrix},
\end{align*}
with $u_1$ as the input and a senor on $\phi$. 
The linearized model $\Sigma_2$ is computed at the origin. 

Figure \ref{fig:enhanced_linearized} presents the numerical results for Case 1, i.e., when a linearized model of a non-linear model is enhanced to operate beyond the linearized regime. From Figure \ref{fig:linearized_output}, the outputs of $\Sigma_1$ and $(T_e,\gamma,\delta)$-$\Sigma_2$ are approximately close, implying that the $(T_e,\gamma,\delta)$-linearized model is a good approximation of its respective FOM.  
Figure \ref{fig:linearized_error} presents the error $y_1 - y_2$ between (1) $\Sigma_1$ and $\Sigma_2$ and (2) $\Sigma_1$ and $(T_e,\gamma,\delta)$-$\Sigma_2$. From Figure \ref{fig:linearized_output}, it follows that the output of $(T_e,\gamma,\delta)$-$\Sigma_2$ closely matches that of $\Sigma_1$. From Figure \ref{fig:linearized_error}, the error between the output of $\Sigma_1$ and $\Sigma_2$ is reduced by enhancing the linearized model to $(T_e,\gamma,\delta)$-linearized model.

\begin{figure}[t]
    \centering
        \centering
        \includegraphics[scale = 0.18]{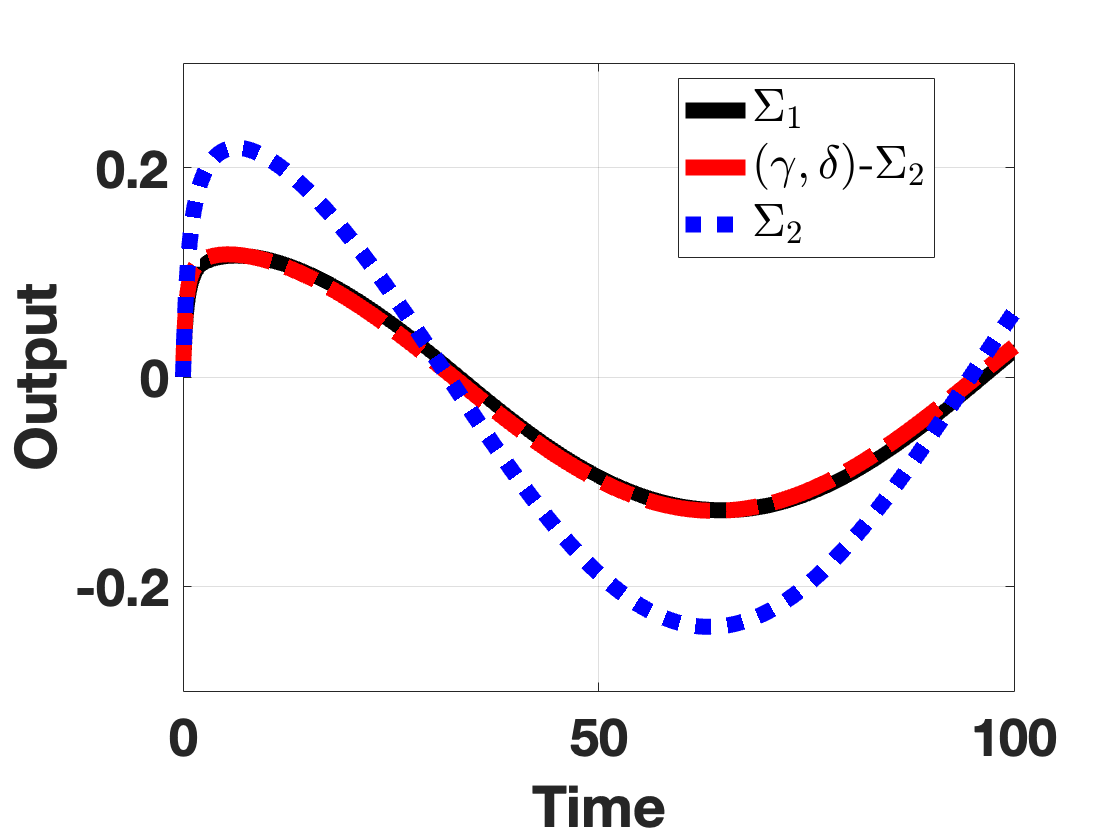}
    \caption{Numerical results for enhancing an ROM for designing digital-twins of reduced complexity (Case 2) in the presence of modeling errors. The black curve denotes the output of FOM $\Sigma_f$, the blue dotted curve denotes the output of the ROM $\Sigma_r$, and the red dashed curve denotes the output of $(T_e,\gamma,\delta)$-$\Sigma_r$.}
    \label{fig:enhanced_rom}
\end{figure}

For the numerical study for Case 2, i.e., to enhance an available ROM of a non-linear FOM for the purpose of designing digital-twins of reduced order, we consider the example of a circuit model analogous to the one considered in \cite{besselink2014model}. Specifically, in the electronic circuit, the resistors and capacitors have
unit resistance and capacitance, respectively. The output $y\in \mathbb{R}$ is taken at the first node and represents the output voltage. 
The nonlinearity is due to the resistor at the first node. This resistor satisfies $i_c = \psi(x_1^1)$ with $x_1^1$ denoting the voltage across the resistor and $i_c$ the current across the resistor. 
The nonlinear model of the circuit is
\begin{align*}
    f_1(x_1) = A_1x_1 + \varphi(x_1) + E_1d_1
\end{align*}
with
\begin{align*}
    A_1 &= \begin{bmatrix}
        -2 & 1 & & 0\\
        1 & -2 & \ddots \\
         & \ddots & \ddots & 1\\
         0 & & 1 & -2
    \end{bmatrix}, 
    \varphi(x_1) = -\begin{bmatrix}
        \varphi(x_1^1)\\
        1\\
        \vdots \\
        1
    \end{bmatrix},
\end{align*}
$\varphi(v) = \sign(v)(v^2)$, and $d_1(t) = \sin(0.2\pi t)$. Here, the state $x_1^{i} \in \mathbb{R}^n$, $i\in \{1, \dots, n\}$ with $n$ set to $100$, represents the voltages at the nodes. 
The nonlinear circuit represents $\Sigma_1$. 
The matrix $E_1$ is chosen as $0.1 A_1$ and represents modeling errors associated with the FOM.
We consider the reduced model $\Sigma_2$, with matrix $A_2$ of order $4$, as determined in \cite{besselink2014model} and select $d_1(t)= x_1(t)$. We compute the $(T_e,\gamma,\delta)$-reduced model, i.e., $(T_e,\gamma,\delta)$-$\Sigma_2$, of $\Sigma_1$ by solving equation \eqref{eq:LMI_W} using an analogous approach as in Case 1 with the bound on $\varphi(x)\in [-3,3]$ and present the numerical results in Figure \ref{fig:enhanced_rom}.

Figure \ref{fig:enhanced_rom} presents the outputs of FOM $\Sigma_1$, ROM $\Sigma_2$, and the $(T_e,\gamma,\delta)$-$\Sigma_2$.
From Figure \ref{fig:enhanced_rom}, it follows that the output of $(T_e,\gamma,\delta)$-$\Sigma_2$ closely matches that of $\Sigma_1$ even under the presence of the modeling errors highlighting the applicability of this framework for designing digital-twins of reduced complexity. 

\begin{remark}
    Many approaches for model reduction of nonlinear systems lack an analytical error bound between the performance of the reduced model and the respective non-linear system. By combining the existing approaches for model reduction with $(T_e,\gamma,\delta)$-similarity framework, these approaches may not only have an improved performance but will also have an analytical error bound characterized in equation \eqref{eq:similarity_def}. 
\end{remark}

In contrast to enhancing a given abstract model of a nonlinear system, one may choose to design an abstract model. In particular, by considering matrices $A_2,B_2,$ and $C_2$ as variables and solving equation \eqref{eq:LMI_W}, it might be possible to determine a $(T_e,\gamma,\delta)$-abstract model for system $\Sigma_1$. This approach has been explored in the case of linear systems in \cite{bajaj2025online,pirastehzad2025hierarchical}. However, this approach may be computationally expensive as equation \eqref{eq:LMI_W} will now be an infinite dimensional bilinear matrix inequality. We also note that this approach will yield the differential dynamics of the $(T_e,\gamma,\delta)$-ROM using which the dynamics of the $(T_e,\gamma,\delta)$-ROM will need to be determined.

\section{Conclusion and Future Works}\label{sec:conclusion}
In this work, we present a notion of system comparison for non-deterministic nonlinear systems via $(T_e,\gamma,\delta)$-similarity. The proposed notion provably satisfies reflexivity and transitivity properties making it suitable for comparing nonlinear systems.
We characterize the notion of $(T_e,\gamma,\delta)$-similarity for nonlinear systems via a feasibility problem consisting of LMIs and provide necessary and sufficient conditions for the same. By introducing the notion of differential $(T_e,\gamma,\delta)$-similarity, we establish equivalence relation between that $(T_e,\gamma,\delta)$-similarity of nonlinear systems  and their differential dynamics. We illustrate the applicability of this framework via two applications. 

One possible limitation of this work is that it requires the information of the state of $\Sigma_1$ to compute $u_d(t)$. In some applications, this may not be feasible or a projection matrix may be required. In this regard, an immediate extension is to the case when only the output of system $\Sigma_1$ is available. 
In the future, we plan to extend the notion of $(T_e,\gamma,\delta)$-similarity to hybrid systems along with its application to complex robotic systems.



\bibliographystyle{ieeetr}
\bibliography{reference}

\end{document}